%% file: 0-main.tex
\documentclass[a4paper,UKenglish,cleveref, autoref, thm-restate]{lipics-v2021}

\title{Partition Constraints for Conjunctive Queries: Bounds and Worst-Case Optimal Joins [TECHNICAL REPORT]}

\author{Kyle {Deeds}}{University of Washington, Seattle, WA, USA}{kdeeds@cs.washington.edu}{}{Was supported by the NSF SHF 2312195 and NSF IIS 2314527 grants}
\author{Timo Camillo {Merkl}}{TU Wien, Vienna, Austria}{timo.merkl@tuwien.ac.at}{}{Was supported by the Vienna Science and Technology Fund 
(WWTF) [10.47379/ICT2201, 10.47379/VRG18013].}

\authorrunning{K. Deeds and T.\,C. Merkl} 

\Copyright{Kyle Deeds and Timo C. Merkl} 

\acknowledgements{Part of this work was done when the authors were visiting the Simons Institute for the Theory of Computing.}

\nolinenumbers 

\EventEditors{John Q. Open and Joan R. Access}
\EventNoEds{2}
\EventLongTitle{42nd Conference on Very Important Topics (CVIT 2016)}
\EventShortTitle{CVIT 2016}
\EventAcronym{CVIT}
\EventYear{2016}
\EventDate{December 24--27, 2016}
\EventLocation{Little Whinging, United Kingdom}
\EventLogo{}
\SeriesVolume{42}
\ArticleNo{23}



\input{macro}
\usetikzlibrary{topaths,calc, backgrounds,fit}
\usepackage{amsmath}

\begin{document}

\maketitle

\begin{abstract}
In the last decade, various works have used statistics on relations to improve both the theory and practice of conjunctive query execution. Starting with the AGM bound which took advantage of relation sizes, later works incorporated statistics like functional dependencies and degree constraints. Each new statistic prompted work along two lines; bounding the size of conjunctive query outputs and worst-case optimal join algorithms. In this work, we continue in this vein by introducing a new statistic called a \emph{partition constraint}.
This statistic captures latent structure within relations by partitioning them into sub-relations which each have much tighter degree constraints. We show that this approach can both refine existing cardinality bounds and improve existing worst-case optimal join algorithms.
\end{abstract}

\section{Introduction}

\input{1-introduction}

\section{Preliminaries}

\input{2-prelim}
\label{sec:prelim}

\section{Partition Constraints}
\label{sec:gdc}

\input{3-gdc}

\section{The Hexagon Query}
\label{sec:hyprtriangle}

\input{4-hypertriangle}

\section{Partition Constraints for General Conjunctive Queries}
\label{sec:general}

\input{5-general}

\section{Conclusions and Future Work}
\label{sec:conclusion}

\input{7-conclusion}

\bibliography{ref}

\appendix
\input{app-4}

\end{document}

%% file: macro.tex
\usepackage{graphicx} 
\usepackage{algorithm}
\usepackage{xspace}
\usepackage{wasysym}
\usepackage[noend]{algpseudocode}
\usepackage{xcolor}
\usepackage{float}
\usepackage[colorinlistoftodos,prependcaption]{todonotes}
\newfloat{algorithm}{t}{lop}

\DeclareMathOperator*{\argmin}{argmin}

\makeatletter
\def\moverlay{\mathpalette\mov@rlay}
\def\mov@rlay#1#2{\leavevmode\vtop{%
   \baselineskip\z@skip \lineskiplimit-\maxdimen
   \ialign{\hfil$\m@th#1##$\hfil\cr#2\crcr}}}
\newcommand{\charfusion}[3][\mathord]{
    #1{\ifx#1\mathop\vphantom{#2}\fi
        \mathpalette\mov@rlay{#2\cr#3}
      }
    \ifx#1\mathop\expandafter\displaylimits\fi}
\makeatother

\newcommand{\cupdot}{\mathbin{{\cup}}}
\newcommand{\bigcupdot}{\mathop{{{\bigcup}}}}

\makeatletter

\newcommand{\hexq}{Q_{\varhexagon}}

\newcommand{\GDC}{PC\xspace}
\newcommand{\GDCs}{PCs\xspace}

%% file: 1-introduction.tex
Efficient query execution is the cornerstone of modern database systems, where the speed at which information is retrieved often determines the effectiveness and user satisfaction of applications. In theoretical work, this problem is often restricted to the enumeration of conjunctive queries (CQs). One of the primary goals of database theory has been to classify the hardness of different queries and provide algorithms that enumerate them in optimal time. Most of these classical guarantees were described relative to the size, $N$, of the largest table in the database. In practical work, there has been a parallel effort to hone query optimizers which automatically generate an algorithm for query evaluation based on the particular properties of the data, finely tailoring the execution to the dataset at hand.

Recently, these lines of work have begun to meaningfully converge in the form of worst-case optimal join (WCOJ) algorithms. 
As a new paradigm in database theory, these methods provide data-dependent guarantees on query execution that take into account statistics $s(D)$ about the dataset being queried $D$. 
Concretely, WCOJ algorithms aim to only require time proportionate to (a bound on) the maximum join size over all databases $D'$ with the same statistics as $D$, i.e., an optimal runtime on the \textit{worst-case instance}. 
Crucially, the statistics that are used directly impact the kinds of guarantees that can be provided; more granular statistics allow for tighter worst-case analyses. 

In a sequence of papers, increasingly more detailed statistics were incorporated into theoretical analyses. The foundational work in this direction by Grohe et al. used the size of each table in the database, often called \textit{cardinality constraints} (CC)s, to produce bounds on the output size \cite{atserias2013size, DBLP:conf/soda/GroheM06}. This eventually led to a multitude of exciting WCOJ algorithms for CCs, e.g., Generic Join \cite{DBLP:journals/sigmod/NgoRR13}, Leapfrog Triejoin \cite{DBLP:conf/icdt/Veldhuizen14}, and NPRR \cite{DBLP:conf/pods/NgoPRR12}.
Instead of processing one join at a time, these algorithms processed one attribute at a time, asymptotically improving on traditional join plans. They have even seen significant uptake in the systems community with a variety of efficient implementations \cite{DBLP:journals/tods/AbergerLTNOR17,DBLP:journals/pvldb/FreitagBSKN20,DBLP:journals/pacmmod/WangWS23}.

In later work, researchers incorporated functional dependencies (FDs) into cardinality bounds and then generalized CCs and FDs to \textit{degree constraints} (DCs) \cite{DBLP:journals/jacm/GottlobLVV12, DBLP:conf/pods/KhamisNS16}. This work culminated in the development of the PANDA algorithm which leveraged information theory and proof-theoretic techniques to provide worst-case optimality relative to the polymatroid bound, modulo an additional poly-logarithmic factor \cite{DBLP:conf/pods/KhamisNS17}. 


A DC for a relation $R$ on the attributes $\textbf{Y}$ and a subset $\textbf{X}\subseteq \textbf{Y}$ asserts that for each fixed instantiation $\textbf{x}$ of $\textbf{X}$ there are only a limited number of completions to the whole set of attributes $\textbf{Y}$. Thus, $\textbf{X}$ functions like a weak version of a key.
However, DCs are a brittle statistic; a single high frequency value per attribute can dramatically loosen a relation's DCs even if all other values only occur once. In the graph setting, where this corresponds to bounded (in- or out-)degree graphs, this has long been viewed as an overly restrictive condition because graphs often have highly skewed degree distributions. To overcome this, the graph theory community identified degeneracy as a natural graph invariant that allows for graphs of unbounded degree while permitting fast algorithms. For example, pattern counting on low degeneracy graphs has recently received considerable attention  \cite{DBLP:journals/jacm/BeraGLSS22, DBLP:conf/innovations/BeraPS20, DBLP:conf/soda/BeraPS21, DBLP:conf/pods/BeraS20, DBLP:conf/focs/0002R21, DBLP:journals/talg/GishbolinerLSY23}.

In this work, we propose partition constraints (\GDCs), a declarative version of graph degeneracy for higher-arity data. \GDCs naturally extend DCs in the sense that every DC can be expressed as \GDC, but not the other way around. Informally, for a \GDC to hold over a relation $R$, we require it to be possible to split $R$ such that each partition satisfies at least one DC. Again, for a binary edge relation $R=E$ of a directed graph $G=(V,E)$, this means that we can partition the edges into two sets $E_1,E_2$, such that the subgraph $G_1:=(V,E_1)$ has bounded out-degree while $G_2:=(V,E_2)$ has bounded in-degree.

We aim to provide a thorough analysis of the effect of \GDCs on conjunctive query answering. To that end, we show that \GDCs can provide asymptotic improvements to query bounds and evaluation, and we demonstrate how they can gracefully incorporate work on cardinality bounds and WCOJ algorithms for DCs. Further, we provide both exact and approximate algorithms for computing \GDCs and inspect to what extent \GDCs are present in well-known benchmark datasets.

\noindent \textbf{Summary of results.}
\begin{itemize}
    \item We introduce \GDCs as a generalization of DCs and provide two algorithms to determine \GDCs in polynomial time as well as to partition the data to witness these constraints.
    One is exact and runs in quadratic time, while the second runs in linear time and provides a constant factor approximation.

    \item We develop new bounds on the output size of a join query. These bounds use DC-based bounds as a black box and naturally extend them to incorporate \GDCs. We show that these bounds are asymptotically tighter than bounds that rely on DCs alone. Further, we show that if the DC-based bounds are tight, then the PC-based bounds built on top of them will be tight as well.
    
    \item Using WCOJ algorithms for DCs as a black box, we provide improved join algorithms that are worst-case-optimal relative to the tighter PC-based bound. Notably, if an algorithm were proposed that is worst-case optimal relative to a tight DC-based bound, then this would immediately result in a WCOJ algorithm relative to a tight PC-based bound.
\end{itemize}

\noindent \textbf{Structure of the paper.}
In Section~\ref{sec:prelim}, we provide some basic definitions and background. We formally introduce and discuss PCs in Section~\ref{sec:gdc} and compute them on common benchmark data. In Section~\ref{sec:hyprtriangle}, we illustrate the benefits of the PC framework by thoroughly analyzing a concrete example. The developed techniques are then extended in Section~\ref{sec:general} and applied to arbitrary CQs. We conclude and give some outlook to future work in Section~\ref{sec:conclusion}. 
Full proofs of some results are deferred to the appendix, and we instead focus on providing intuition in the main body.

%% file: 2-prelim.tex

\textit{Conjunctive queries.}
We assume the reader is familiar with relational algebra, in particular with joins, projections, and selection, which will respectively be denoted by $\Join, \pi,$ and $\sigma$.
In the context of the present paper, a \textit{(conjunctive) query} (CQ), denoted using relational algebra, is an expression of the form
\begin{align*}
    Q(\mathbf{Z}) \leftarrow R_1(\mathbf{Z}_1)\Join \dots \Join R_k(\mathbf{Z}_k).
\end{align*}
In this expression, each $R_i(\mathbf{Z}_i)$ is a relation over the set of variables $\mathbf{Z}_i$, and $Q(\mathbf{Z})$ is the output of the query where $\mathbf{Z}= \bigcup_i \mathbf{Z}_i$.
Thus, we only consider \textit{full conjunctive query}.
When clear from the context, we omit the reference to the set of variables and simply write $R_i$ and $Q$.
We also use $Q$ to refer to the whole CQ.
A database instance $I$ for $Q$ is comprised of a concrete instance for each $R_i$, i.e., a set of tuples (also denoted by $R_i$) where each tuple $\textbf{z}_i$ contains a value for each variable in $\mathbf{Z}_i$. 
The set of values appearing in $I$ is referred to as $dom$, the domain of $I$.
Furthermore, let $dom(Z)$ be the values assigned to $Z\in \textbf{Z}$ by some relation $R_i$ and, for $\textbf{Y}\subseteq \textbf{Z}$, $dom(\textbf{Y}):= \times_{Y\in \textbf{Y}}dom(Y)$.
For a particular database instance $I$, we denote the answers to a CQ, i.e., the join of the $R_i$, as a relation $Q^I(\mathbf{Z})$. 
Lastly, a join algorithm $\mathcal{A}$ is any algorithm that receives a query $Q$ and a database instance $I$ as input and outputs the relation $Q^I(\mathbf{Z})$.


\textit{Degree constraints and bounds.}
In recent years, there has been a series of novel approaches to bound the size of $Q^I(\mathbf{Z})$ based on the structure of the query $Q$ and a set of statistics about the database instance $I$. For full conjunctive queries, this recent round of work began with the AGM bound \cite{atserias2013size}. This bound takes the size of each input relation as the statistics and connects the size of the result to a weighted edge covering of the hypergraph induced by $Q$. Later bounds extended this approach by considering more complex statistics about the input relations. Functional dependencies were investigated first in \cite{DBLP:journals/jacm/GottlobLVV12}. These were then generalized to degree constraints (DCs) in \cite{DBLP:conf/pods/KhamisNS16} and degree sequences (DSs) in \cite{DBLP:conf/icdt/DeedsSBC23,DBLP:journals/pacmmod/DeedsSB23,DBLP:journals/pacmmod/KhamisNOS24}. Our work in this paper continues in this vein by generalizing degree constraints further to account for the benefits of partitioning relations. 
So, we begin by defining degree constraints below.
\begin{definition}
Fix a particular relation $R({\bf Z})$.
Given two sets of variables $\mathbf{X}, \mathbf{Y}$ where $\mathbf{X}\subseteq\mathbf{Y}\subseteq \mathbf{Z}$, a degree constraint of the form $DC_R({\bf X}, {\bf Y}, d)$ implies the following, 
\begin{align*}
    \max_{{\bf x}\in dom(\bf X)}|\pi_{\mathbf{Y}}\sigma_{\mathbf{X} = \mathbf{x}}R| \leq d.
\end{align*}
A database instance $I$ satisfying a (set of) degree constraints $\mathbf{DC}$ is denoted by $I \vDash \mathbf{DC}$.
For convenience, we denote the minimal $d$ such that $DC_R({\bf X}, {\bf Y}, d)$ holds by $DC_R(\mathbf{X},\mathbf{Y})$.
If $\mathbf{Y}=\mathbf{Z}$ and $\mathbf{X} = \emptyset$ the constraint simply bounds the cardinality of $R$ and we write $CC_R(d)$.
If clear from the context, we may omit $R$.
\label{def:degree-constraint}
\end{definition}

In addition to generalizing DCs, our work is able to refine any of the previous bounding methods for DCs by incorporating them into our framework. Therefore, we introduce a general notation for DC-based bounds.
\begin{definition} 
\label{def:cardinality-bound}    
Given a conjunctive query $Q$ and a set of degree constraints $\mathbf{DC}$, a cardinality bound $CB(Q,\mathbf{DC})$ is any function where the following is true,
\begin{align*}
    |Q^I(\mathbf{Z})| \leq CB(Q,\mathbf{DC}) \quad\forall \,\, I\vDash \mathbf{DC}.
\end{align*}
\end{definition}


Throughout this work, we will make specific reference to the combinatorics bound which we describe below. While it is impractical, this bound is computable and tight which makes it useful for theoretical analyses.

\begin{definition}
Given a conjunctive query $Q$ and a set of degree constraints $\mathbf{DC}$, we define the combinatorics bound (for degree constraints) as 
\begin{align*}
    CB_{\text{Comb}}(Q,\mathbf{DC}) = \sup_{I\vDash\mathbf{DC}}|Q^I(\mathbf{Z})|.
\end{align*}
\end{definition}
In this work, we make two minimal assumptions about bounds and statistics. First, we assume that cardinality bounds are finite by assuming that every variable in the query is covered by at least one cardinality constraint. Second, we assume that there is a bound on the growth of the bounds as our statistics increase. Formally, for a fixed query $Q$, we assume that every cardinality bound $CB$ has a function $f_Q$ such that for any $\alpha \in \mathbb{R}^+$ and degree constraints $\textbf{DC}$ we have $CB(Q,\alpha \cdot\textbf{DC})\leq f_Q(\alpha )\cdot CB(Q,\textbf{DC})$. Usually, cardinality bounds are expressed in terms of power products of the degree constraints \cite{atserias2013size,DBLP:journals/jacm/GottlobLVV12,DBLP:journals/pacmmod/KhamisNOS24}. In that case, $f_Q=O(\alpha ^c)$ for some constant $c$.

Lastly, we will also incorporate existing work on worst-case optimal join (WCOJ) algorithms. Each WCOJ algorithm is optimal relative to a particular cardinality bound, and we describe that relationship formally as follows:
\begin{definition}
\label{def:worst-case-optimal}
    Denote the runtime of a join algorithm $\mathcal{A}$ on a conjunctive query $Q$ and database instance $I$ as $\mathcal{T}(\mathcal{A}, Q, I)$. 
    $\mathcal{A}$ is worst-case optimal (WCO) relative to a cardinality bound $CB_{\mathcal{A}}$ if the following is true for all queries $Q$,
    \begin{align*}
        \mathcal{T}(\mathcal{A}, Q, I) = O(|I| + CB_{\mathcal{A}}(Q, \mathbf{DC})), \quad \forall \,\,\mathbf{DC}, \,\,I\vDash\mathbf{DC}.
    \end{align*}
    Note that the hidden constant may only depend on the query $Q$ and, importantly, neither on the set of degree constraints $\mathbf{DC}$ nor on the database instance $I$.
    When $\mathcal{A}$ is optimal relative to $CB_{\text{Comb}}$, we simply call it worst-case optimal (relative to degree constraints).
\end{definition}
Note that some algorithms fulfill a slightly looser definition of worst-case optimal by allowing additional poly-logarithmic factors \cite{DBLP:conf/pods/KhamisNS17}. These algorithms can still be incorporated into this framework, but we choose the stricter definition to emphasize that we do not induce additional factors of this sort.

Much of the recent work on WCOJ algorithms can be categorized as \emph{variable-at-a-time} (VAAT) algorithms. Intuitively, a VAAT computes the answers to a join query by answering the query for an increasingly large number of variables. 
This idea is at the heart of Generic Join \cite{DBLP:journals/sigmod/NgoRR13}, Leapfrog Triejoin \cite{DBLP:conf/icdt/Veldhuizen14}, and NPRR \cite{DBLP:conf/pods/NgoPRR12}. 
We define this category formally as follows. 
For an arbitrary query $Q(\textbf{Z})\leftarrow R_1 \Join \dots \Join R_k$, let $Q_{i}$ denote the sub-query
$$Q_{i}(Z_1,\dots Z_i) \leftarrow \pi_{Z_1,\dots,Z_i} R_1 \Join \dots \Join \pi_{Z_1,\dots,Z_i} R_k,$$
for an ordering $Z_1,\dots,Z_r = \textbf{Z}$.

\begin{definition}
\label{def:vaat}
    A join algorithm $\mathcal{A}$ that solves conjunctive queries $Q(\mathbf{Z}) \leftarrow R_1(\mathbf{Z}_1) \Join \dots \Join R_k(\mathbf{Z}_k)$ is a VAAT algorithm if there is some ordering $Z_1,\dots,Z_r = \mathbf{Z}$ such that $\mathcal{A}$ takes time $\Omega(\max_i|Q_{i}^I|)$ for databases $I$.
    The choice of the ordering is allowed to depend on $I$.
\end{definition}

%% file: 3-gdc.tex
Partition constraints (PCs) extend the concept of DCs by allowing for the partitioning of relations. For clarity, we start with the binary setting, revisiting the example from the introduction. To that end, let $E(X,Y)$ be the edge relation of a directed graph $G=(V,E)$. For a degree constraint to hold on $E$, either the in- or the out-degree of $E$ must be bounded. However, this is a strong requirement and often may not be satisfied, e.g. on a highly skewed social network graph. In these cases, it makes sense to look for further latent structure in the data. We suggest partitioning $E$ such that each part satisfies a (different) degree constraint. We are interested in the following quantity where the minimum is taken over all bi-partitions of $E$, i.e. $E = E^X\cupdot E^Y$ (implicitly $E^X\cap E^Y=\emptyset$),
\begin{align}
    \min_{E^{X}\cupdot E^{Y} = E}\max\{DC_{E^X}(X, XY), DC_{E^Y}(Y, XY)\}.
    \label{eq:quantity}
\end{align}
This corresponds to a splitting of the graph into two graphs.  In one of them, $G^X=(V,E^X)$, we attempt to minimize the maximum out-degree of the graph while in the other, $G^Y=(V,E^Y)$, we attempt to minimize the maximum in-degree.  It is important to note that both of these quantities can be arbitrarily lower than the maximum in-degree and out-degree of the original graph.  This is where the benefit of considering partitions comes from.

An example of such a partitioning is depicted in Figure \ref{fig:graphExample}. There, the dashed blue edges represent $E^X$ while the solid red edges represent $E^Y$. The maximum out- and in-degree of the full graph is 5 while $G^X=(V,E^X)$ has a maximum out-degree of 1 and $G^Y=(V,E^Y)$ has a maximum in-degree of 1. In general, a class of graphs can have an unbounded out- and in-degree but always admit a bi-partitioning with an out- and in-degree of 1, respectively.
    \begin{figure}
    \begin{minipage}[b]{0.33\textwidth}
    \centering
    \resizebox{\textwidth}{!}{
    \begin{tikzpicture}
        \foreach \angle/\label in {112.5/A, 67.5/A2, 22.5/B, -22.5/C, -67.5/D, -112.5/D2, -157.5/E, 157.5/F}
        {
            \node (\label) at (\angle:1.2) { $\bullet$};
        }
        \draw [->, draw=red](A) -- (B);
        \draw [->, draw=red](A) -- (C);
        \draw [->, draw=red](A) -- (E);
        
        \draw [->, draw=red](B) -- (A);
        \draw [->, draw=red](B) -- (F);
        \draw [->, draw=red](B) -- (D);
        \draw [->, draw=red](A) -- (A2);
        \draw [->, draw=red](E) -- (D2);

        \draw [->, dash pattern=on 1pt off 1pt, draw=blue](A) -- (F);
        \draw [->, dash pattern=on 1pt off 1pt, draw=blue](B) -- (E);
        \draw [->, dash pattern=on 1pt off 1pt, draw=blue](C) -- (D);
        \draw [->, dash pattern=on 1pt off 1pt, draw=blue](D) -- (F);
        \draw [->, dash pattern=on 1pt off 1pt, draw=blue](E) -- (F);
        \draw [->, dash pattern=on 1pt off 1pt, draw=blue](F) -- (C);
        \draw [->, dash pattern=on 1pt off 1pt, draw=blue](A2) -- (D2);
        \draw [->, dash pattern=on 1pt off 1pt, draw=blue](D2) -- (F);
    \end{tikzpicture}
    }
        \caption{Graph Example.}
        \label{fig:graphExample}

    \end{minipage}
    \hfill
    \begin{minipage}[b]{0.6\textwidth}
        \small
        \begin{tabular}{|l|l|}
        \multicolumn{2}{c}{\texttt{Access}} \\ \hline
        \texttt{PersonID}     &  \texttt{RoomID} \\ \hline
         Ava    & Beacon Hall \\ \hline
         Ben    & Beacon Hall \\ \hline
         Cole    & Delta Hall \\ \hline
         Dan    & Delta Hall \\ \hline
         Emma    & Gala Hall \\ \hline
         Finn    & Jade Hall \\ \hline
         Porter    & Beacon Hall \\ \hline
         Porter    & Delta Hall \\ \hline
         Porter    & Gala Hall \\ \hline
         Porter    & Jade Hall \\ \hline
        \end{tabular}
        \hfill
        \begin{tabular}{|l|l|}
        \multicolumn{2}{c}{$\texttt{Access}^\texttt{PersonID}$} \\ \hline
        \texttt{PersonID}     &  \texttt{RoomID} \\ \hline
         Ava    & Beacon Hall \\ \hline
         Ben    & Beacon Hall \\ \hline
         Cole    & Delta Hall \\ \hline
         Dan    & Delta Hall \\ \hline
         Emma    & Gala Hall \\ \hline
         Finn    & Jade Hall \\ \hline 
         \multicolumn{2}{c}{} \\
        \multicolumn{2}{c}{$\texttt{Access}^\texttt{RoomID}$} \\ \hline
        \texttt{PersonID}     &  \texttt{RoomID} \\ \hline         
         Porter    & Beacon Hall \\ \hline
         Porter    & Delta Hall \\ \hline
         Porter    & Gala Hall \\ \hline
         Porter    & Jade Hall \\ \hline
        \end{tabular}
        \caption{Access Example.}
        \label{fig:access}
    \end{minipage}
    \end{figure}

This approach is originally motivated by the graph property \textit{degeneracy}. Intuitively, this property tries to allocate each edge of an undirected graph to one of its incident vertices such that no vertex is ``responsible'' for too many edges. Each partition $E = E^X\cupdot E ^Y$ can be seen as a possible allocation where the edges in $E ^X$ are allocated to the $X$ part of the tuples while the edges in $E^Y$ are allocated to the $Y$ part of the tuples. 
Thus, in general, if the undirected version of a graph class $\mathcal{G}$ has bounded degeneracy, the Quantity \ref{eq:quantity} must also be bounded for $\mathcal{G}$.
In fact, the converse is true as well if the domain of the $X$ and $Y$ attributes are disjoint.

Formally, we define a \GDC on a relation $R$ as below. Note, we say a collection of subrelations $(R^1,\dots, R^k)$ partition $R$ when they are pairwise disjoint and $\bigcup_j R^j=R$.
\begin{definition}
Fix a particular relation $R(\mathbf{Z})$, a subset $\mathbf{Y}\subseteq \mathbf{Z}$, and let $\mathcal{X} = \{\mathbf{X}^1,...,\mathbf{X}^{|\mathcal{X}|}\} \subseteq 2^\mathbf{Y}$ be a set of sets of variables. 
Then, a partition constraint of the form $\GDC_R(\mathcal{X}, {\bf Y}, d)$ implies,
    \begin{align*}
    \min_{(R^\mathbf{X})_{\mathbf{X}\in \mathcal{X}} \text{ partition } R} \max \{DC_{R^\mathbf{X}}(\mathbf{X}, {\bf Y})\mid \mathbf{X} \in \mathcal{X}\} \leq d.
    \end{align*}
Again, a database instance $I$ satisfying a (set of) partition constraints $\mathbf{\GDC}$ is denoted by $I \vDash \mathbf{\GDC}$.
We denote the minimal $d$ such that $\GDC_R(\mathcal{X}, {\bf Y}, d)$ holds by $\GDC_R(\mathcal{X},\mathbf{Y})$
and we omit $R$ if clear from the context.
\end{definition}


This definition says that one can split the relation $R$ into disjoint subsets $R^j \subseteq R, \bigcup_j R^j = R$ and associate each $R^j$ with a degree constraint over a set of variables $\textbf{X}^j\in \mathcal{X}$ such that the maximum of these constraints is then bounded by $d$. From an algorithmic point of view, each part $R^j$ should be handled differently to make use of its unique, tighter DC. The core of this work is in describing how these partitions can be computed and how to make use of the new constraints on each part. Broadly, we show how these new constraints produce tighter bounds on the join size and how algorithms can meet these bounds. Note \GDCs are a strict generalization of DCs since we can define an arbitrary degree constraint $DC(\mathbf{X}, \mathbf{Y}, d)$ as $\GDC(\{\mathbf{X}\}, \mathbf{Y}, d)$. For simplicity, when we discuss a collection of \GDCs, we assume that there is at most one \GDCs per $(\mathcal{X}, \textbf{Y})$ pair, i.e., there are never two $\GDC_R(\mathcal{X}, \mathbf{Y}, d), \GDC_R(\mathcal{X}, \mathbf{Y}, d'), d\neq d',$ as it suffices to keep the stronger constraint.

Next, we present an example of how relations with small \GDCs might arise in applications.

\begin{example}
    Consider a relation $\texttt{Access(PersonID, RoomID)}$ that records who has access to which room at a university.  This could be used to control the key card access of all faculty members, students, security, and cleaning personnel. Most people (faculty members and students) only need access to a limited number of rooms (lecture halls and offices). On the other hand, porters and other caretaker personnel need access to many different rooms, possibly all of them.  However, each room only needs a small number of people taking care of it.  Thus, it makes sense to partition $\texttt{Access}$ into $\texttt{Access}^\texttt{PersonID}\cup \texttt{Access}^\texttt{RoomID}$ with the former tracking the access restrictions of the faculty members and students, and the latter tracking the access restrictions of the caretaker personnel. With this partition, both $DC_{\texttt{Access}^\texttt{PersonID}}(\texttt{PersonID}, \texttt{RoomID})$ and $DC_{\texttt{Access}^\texttt{RoomID}}(\texttt{RoomID}, \texttt{PersonID})$ should be small. 
    Thus, $\GDC_{\texttt{Access}}(\{\texttt{PersonID}, \texttt{RoomID}\}, \texttt{PersonID RoomID})$ should be small as well.

    Figure \ref{fig:access} depicts a small example instance of the situation. There, the students each only need access to a single lecture hall to attend their courses, and all the rooms are taken care of by a single porter. Thus, following the suggested splitting, the respective DC for both subrelations $\texttt{Access}^\texttt{PersonID}$ and $\texttt{Access}^\texttt{RoomID}$ is 1 and, therefore, also the PC for the whole relation \texttt{Access} is 1. 
\end{example}

To see how these statistics manifest in real world data, we calculated the \GDC of relations from some standard benchmarks which are displayed in Figure~\ref{table:gdcs} \cite{DBLP:conf/sigmod/Sun020,DBLP:journals/pvldb/HanWWZYTZCQPQZL21,DBLP:journals/pvldb/LeisGMBK015}. We computed the \GDC using Algorithm \ref{alg:exactdecomp} from Section \ref{sec:general}. To model the interesting case of many-to-many joins, we first removed any attributes which are primary keys from each relation. Specifically, we computed the partition constraint $\GDC(\{X\mid X\in \mathbf{Y}\}, \mathbf{Y})$ where $\mathbf{Y}$ is the set of non-key attributes. 
We compare this with the minimum and maximum degree of these attributes before partitioning. Naively, by partitioning the data randomly, one would expect a \GDC roughly equal to the maximum DC divided by $|\mathcal{X}|$. Alternatively, by placing all tuples in the partition corresponding to the minimum DC, one can achieve a \GDC equal to the minimum DC. However, the computed \GDC is often much lower than both of these quantities. This implies that the partitioning is uncovering useful structure in the data rather than simply distributing high-degree values over multiple partitions.

\begin{figure}
\begin{minipage}[b]{0.6\textwidth}
\small
\begin{tabular}{|l|r|r|r|r|}
\hline
\textbf{Dataset}      & \textbf{Max DC} & \textbf{Min DC} & \textbf{\GDC} & $|\mathcal{X}|$ \\ \hline
aids                  & 11              & 11              & 3             & 2              \\ \hline
yeast                 & 154             & 119             & 9             & 2              \\ \hline
dblp                  & 321             & 113             & 34            & 2              \\ \hline
wordnet               & 526             & 284             & 3             & 2              \\ \hline
Stats/badges          & 899             & 456             & 8             & 2              \\ \hline
Stats/comments        & 134887          & 45              & 15            & 3              \\ \hline
Stats/post\_links     & 10186           & 13              & 2             & 4              \\ \hline
Stats/post\_history   & 91976           & 32              & 3             & 4              \\ \hline
Stats/votes           & 326320          & 427             & 33            & 4              \\ \hline
IMDB/keywords         & 72496           & 540             & 71            & 2              \\ \hline
IMDB/companies        & 1334883         & 94              & 13            & 4              \\ \hline
IMDB/info             & 13398741        & 2937            & 123           & 4              \\ \hline
IMDB/cast             & 25459763        & 1741            & 52            & 6              \\ \hline
\end{tabular}
\caption{Example \GDCs.}
\label{table:gdcs}
\end{minipage}
\hfill
\begin{minipage}[b]{0.35\textwidth}
\centering
\resizebox{.8\textwidth}{!}{
\begin{tikzpicture}
    \foreach \angle/\label in {90/A, 30/W, -30/B, -90/U, -150/C, 150/V}
    {
        \node (\label) at (\angle:1.5) {\Large $\label$};
    }
    \draw [rounded corners=.9cm] (110:2.3) -- (30:2.3) -- (-50:2.3) -- cycle;
    \draw [rounded corners=.9cm] (70:2.3) -- (150:2.3) -- (230:2.3) -- cycle;
    \draw [rounded corners=.9cm] (190:2.3) -- (270:2.3) -- (-10:2.3) -- cycle;
    \draw [rounded corners=.9cm] (30:2.3) -- (150:2.3) -- (270:2.3) -- cycle;
\end{tikzpicture}
}
    \caption{The Query $\hexq$.}
    \label{fig:hypertriangle}
\end{minipage}

\end{figure}

\subsection{Further Partitioning}
At this point, one might wonder if further partitioning the data can meaningfully reduce a \GDC. That is, whether for a given relation $R(\textbf{Z})$ and a particular set of degree constraints $\{DC_R(\textbf{X},\textbf{Y})\mid \textbf{X}\in \mathcal{X}\}$, is it possible to decrease the maximum DC significantly by partitioning $R$ into more than $|\mathcal{X}|$ parts?  We show that this is not the case; neither pre-partitioning nor post-partitioning the data into $k$ parts can reduce the \GDC by more than a factor of $k$. We prove the former first.

\begin{restatable}{proposition}{thmpre}
\label{thm:pre}
    Given a relation $R(\mathbf{Z})$, a subset $\mathbf{Y}\subseteq \mathbf{Z}$, variable sets $\mathcal{X} \subseteq 2^\mathbf{Y}$, and subrelations $(R^1,\dots,R^k)$ that partition $R$. 
    Then,
    \begin{align*}
        \max_{i=1,\dots,k} \GDC_{R^i}(\mathcal{X},\mathbf{Y})\geq \GDC_{R}(\mathcal{X},\mathbf{Y})/k.
    \end{align*}
\end{restatable}
\begin{proof}
For the sake of contradiction, we assume that there exists a partitioning $(R^1,\dots,R^k)$ of $R$ such that,
    \begin{align*}
        \max_{i=1,\dots,k}\GDC_{R^i}(\mathcal{X},\mathbf{Y})< \GDC_{R}(\mathcal{X},\mathbf{Y})/k.
    \end{align*}
Then, we can partition each $R^i$ to witness $\GDC_{R^i}(\mathcal{X},\mathbf{Y})$. Let $\bigcup_{\textbf{X}\in \mathcal{X}} R^{i,\textbf{X}}=R^i$ be such that $DC_{R^{i,\textbf{X}}}(\textbf{X}, \textbf{Y})\leq \GDC_{R^i}(\mathcal{X},\mathbf{Y})$ holds for each part $R^{i,\textbf{X}}$. For each fixed $\textbf{X}\in \mathcal{X}$, we can combine the sub-relations $R^{1,\textbf{X}}, \dots, R^{k,\textbf{X}}$ into a single relation $R^\textbf{X}$. These relations, $(R^\textbf{X})_{\textbf{X}\in \mathcal{X}}$, form a partition of $R$.
The DC for each $R^\textbf{X}$ is at most the sum of the DCs of $R^{1,\textbf{X}}, \dots, R^{k,\textbf{X}}$. Further, this sum must be less than our initial \GDC by our assumption,
\begin{align*}
    DC_{R^{\textbf{X}}}(\textbf{X},\textbf{Y}) \leq k\cdot \GDC_{R^i}(\mathcal{X},\mathbf{Y})<\GDC_{R}(\mathcal{X},\mathbf{Y}).
\end{align*}
This directly implies that
\begin{align*}
    \max_{\mathbf{X}\in\mathcal{X}}DC_{R^{\mathbf{X}}}(\mathbf{X},\textbf{Y}) < \GDC_{R}(\mathcal{X},\mathbf{Y}).
\end{align*}
Because the PC is defined as the minimum value of this maximum DC over all possible partitions of $R$ into $|\mathcal{X}|$ parts, this is a contradiction.
\end{proof}
Next, we show that post-partitioning cannot super-linearly reduce the degree either.
\begin{proposition}
\label{thm:post}
    Let $(R^1,\dots,R^k)$ partition a relation $R(\mathbf{Z})$, and let $\mathbf{X}\subseteq \mathbf{Y}\subseteq \mathbf{Z}$ be two subsets of variables. Then,    \begin{align*}
        \max_{i=1,\dots,k} DC_{R^i}(\mathbf{X}, \mathbf{Y}) \geq DC_{R}(\mathbf{X}, \mathbf{Y})/k.
    \end{align*}
\end{proposition}
\begin{proof}
    Let $\mathbf{x}$ be the value of $\mathbf{X}$ within $R$ that occurs in tuples with $d=DC_{R}(\mathbf{X},\mathbf{Y})$ unique values $\textbf{y}$ of $\mathbf{Y}$. 
    For each of these values $\mathbf{y}$, a tuple containing $\mathbf{y}$ must be associated with one partition $R^i$.
    By the pigeonhole principle and the fact that there are $DC_{R}(\mathbf{X}, \mathbf{Y})$ of these tuples, it follows that some partition must contain at least $DC_{R}(\mathbf{X}, \mathbf{Y})/k$ of these tuples. 
    Therefore, for some $i\in \{1,\dots,k\}$, we have $DC_{P_i}(\mathbf{X}, \mathbf{Y})\geq DC_{R}(\mathbf{X}, \mathbf{Y})/k$.
\end{proof}

Theorems \ref{thm:pre} and \ref{thm:post} together show that for a given relation ${R}(\textbf{Y})$ and a particular set of degree constraints $\{DC_R(\textbf{X},\textbf{Z})\mid \textbf{X}\in \mathcal{X}\}$ deemed relevant, we only have to consider partitionings of $R$ into at most $|\mathcal{X}|$ pieces. 
Thus, if the query size is viewed as a constant, then the number of useful partitions is also a constant.




%% file: 4-hypertriangle.tex
In this section, we show the benefits of the \GDC framework by demonstrating how it can lead to asymptotic improvements for both cardinality bounds and conjunctive query evaluation. Concretely, there is an example query and class of database instances where bounds based on \textbf{\GDC} are asymptotically tighter than those based on \textbf{DC}. Further, all VAAT algorithms (See Definition~\ref{def:vaat}) are asymptotically slower than a \textbf{\GDC}-aware algorithm on this query and instance class. 
Formally, our aim is to show the following:
\begin{theorem}
    \label{thm:hex_query}
    There exists a query $Q$ and a set of partition constraints $\mathbf{\GDC}$ with the degree constraint subset $\mathbf{DC}\subset\mathbf{\GDC}$ such that the following are true:
    \begin{enumerate}
        \item Bounds based on degree constraints are asymptotically sub-optimal, i.e.
        \begin{align*}
            \sup_{I\vDash\mathbf{PC}}|Q^I(\mathbf{Z})| = o(CB_{Comb}(Q,\mathbf{DC})) = o(\sup_{I\vDash\mathbf{DC}}|Q^I(\mathbf{Z})|).
        \end{align*}
        \item There is an algorithm that enumerates $Q^I$ for instances $I\vDash\mathbf{\GDC}$ in time $O(\sup_{I\vDash\mathbf{PC}}|Q^I(\mathbf{Z})|)$.
        \item No VAAT algorithms can enumerate $Q^I$ for instances $I\vDash\mathbf{\GDC}$ in time $O(\sup_{I\vDash\mathbf{PC}}|Q^I(\mathbf{Z})|)$.
    \end{enumerate}
\end{theorem}

\noindent To prove these claims, we consider the hexagon query (also depicted in Figure \ref{fig:hypertriangle})
{
\begin{align*}
    \hexq(A,B,C,U,V,W)\leftarrow R_1(A,W,B)\Join R_2(B,U,C)\Join R_3(C,V,A)\Join R_4(U,V,W)
\end{align*}
}
and impose the following set of \GDCs on the relations:
\begin{alignat*}{10}
&DC_{R_1}(AW,AWB,1), \quad  &&DC_{R_1}(WB,AWB,1), \quad &&CC_{R_1}(n), \quad CC_{R_2}(n) \\ 
&DC_{R_2}(BU,BUC,1), \quad  &&DC_{R_2}(UC,BUC,1), \quad  &&CC_{R_3}(n), \quad  CC_{R_4}(n), \\
&DC_{R_3}(CV,CVA,1), \quad   &&DC_{R_3}(VA,CVA,1),  \quad &&\GDC_{R_4}(\{U,V,W\}, UVW, 1).
\end{alignat*}
We denote the whole set as $\textbf{\GDC}$ and the subset of DCs as \textbf{DC}. 
We now prove the theorem's first claim. 
That is, the combinatorics bound on $\hexq$ is super linear when only considering \textbf{DC}, while there are only a linear number of answers to $\hexq$ over any database satisfying \textbf{\GDC}. We begin by providing a lower bound on the combinatorics bound of $\hexq$.

\begin{restatable}{lemma}{lemCBsuperlinear}
\label{lem:CBsuperlinear}
    The combinatorics bound of $\hexq$ based on $\mathbf{DC}$ is in $\Omega(n^{\frac{4}{3}})$.
\end{restatable}

\begin{proof}[Proof Sketch.]
    It suffices to provide a collection of databases $\mathcal{I}$ such that $I\vDash \textbf{DC}$ and $|\hexq^I| = \Omega(n^\frac{4}{3})$ for $I\in \mathcal{I}$.
    To accomplish this, we introduce a new relation $R_{X,Y,Z}(X,Y,Z)$ with $|dom(X)| = |dom(Z)| = n^{\frac{2}{3}}$ and $|dom(Y)| = n^{\frac{1}{3}}$.
    Intuitively, think of $R_{X,Y,Z}$ as a bipartite graph from the domain of $X$ to the domain of $Z$ where $Y$ identifies the edge for a given $x\in dom(X)$ or $z\in dom(Z)$.
    Thus, every $x\in dom(X)$ is connected to $n^{\frac{1}{3}}$ neighbors in $dom(Z)$ and, due to symmetry, also the other way around.

    Therefore, the constraints $DC(XY,XYZ,1)$, $DC(YZ, XYZ, 1)$, and $CC(n)$ are satisfied over $R_{X,Y,Z}$.
    Thus, we can use a relation of this type for $R_1, R_2, R_3$.
    Concretely, we use this relation in the following way:
        $R_1 = R_{A,W,B}, 
        R_2 = R_{B,U,C},
        R_3 = R_{C,V,A}.$
    Thus, for each ($n^{\frac{2}{3}}$ many) $a\in dom(A)$ there are $n^{\frac{1}{3}}$ many matching $b\in dom(B)$ and possibly up to $n^{\frac{1}{3}}$ many matching $c\in dom(C)$.
    Thus, in total, there may be up to $n^{\frac{4}{3}}$ answers to $R_1\Join R_2 \Join R_3$.
    To accomplish this also for $\hexq$, we only have to make sure that the variables $U, V, W$ also join in $R_4$.
    For that, we simply set $R_4 = dom(U)\times dom(V)\times dom(W)$.
    This relation then satisfies its cardinality constraint.
    (Note that it does not satisfy its \GDC.)    
\end{proof}

We now provide an algorithm (Algorithm~\ref{alg:hypertriangle}) that enumerates $\hexq$ in linear time for databases with the \GDC on $R_4$. 
This proves that the output size is linear, simultaneously completing our proof of claim 1 and claim 2. 
Algorithm~\ref{alg:hypertriangle} first decomposes $R_4$ into three parts with one DC on each part. 
For this, we apply a linear time greedy partitioning algorithm (for more details see Algorithm~\ref{alg:approxdecomp}) and show that it results in partitions whose constraints are within a factor of 3 from the optimal \textbf{\GDC}. 
Then, for each part, a variable order is selected to take advantage of all DCs. 
As an example, for the part $R_{4}^{W}(U,V,W)$, there are at most 3 matching $(u,v)\in dom(U,V)$ pair for each value of $w\in dom(W)$. 
Thus, starting with a tuple $(a,w,b)$ of $R_1$, we can use this fact to determine these values for $u,v$ and then combine this with $DC_{R_2}(BU, BUC,1)$ to determine the unique value for $c$. 
By this reasoning, all of the inner for-loops iterate over a single element or 3 pairs of elements. 
Thus, the nested loops are linear in their total runtime. 
We defer the formal proof to Appendix~\ref{app:hyprtriangle}.

\begin{algorithm}
    \begin{algorithmic}[1] 
        \State $R_{4}^{U}, R_{4}^{V}, R_{4}^{W} \gets \textbf{decompose}(R_4, \{U,V,W\}, UVW)$       
        \State \Comment{$DC_{R_{4}^{U}}(U,UVW, 3), DC_{R_{4}^{V}}(V,UVW, 3), DC_{R_{4}^{W}}(W,UVW, 3)$}
        \For{$(a,w,b)\in R_1$}
            \For{$(u,v)\in \pi_{U,V}\sigma_{W=w}R_{4}^{W}$}
                \For{$c\in (\pi_{C}\sigma_{B=b \land U=u}R_2 \cap \pi_{C}\sigma_{A=a \land V=v}R_3)$}
                    \State \textbf{output} $(a,b,c,u,v,w)$
                \EndFor
            \EndFor
        \EndFor
        \For{$(b,u,c)\in R_2$}
            \For{$(v,w)\in \pi_{V,W}\sigma_{U=u}R_{4}^{U}$}
                \For{$a\in (\pi_{A}\sigma_{B=b \land W=w}R_1 \cap \pi_{A}\sigma_{C=c \land V=v}R_3)$}
                    \State \textbf{output} $(a,b,c,u,v,w)$
                \EndFor
            \EndFor
        \EndFor
        \For{$(c,v,a)\in R_3$}
            \For{$(u,w)\in \pi_{U,W}\sigma_{V=v}R_{4}^{V}$}
                \For{$b\in (\pi_{B}\sigma_{A=a \land W=w}R_1 \cap \pi_{B}\sigma_{C=c \land U=u}R_2)$}
                    \State \textbf{output} $(a,b,c,u,v,w)$
                \EndFor
            \EndFor
        \EndFor
    \end{algorithmic}
    \caption{Linear Hexagon Algorithm}
    \label{alg:hypertriangle}
\end{algorithm}

\begin{restatable}{lemma}{lemalgohyperlinear}    
\label{lem:algohyperlinear}
    Algorithm \ref{alg:hypertriangle} enumerates $\hexq^I$ in time $O(n)$ for databases $I\vDash\mathbf{PC}$.
\end{restatable}

Lemma \ref{lem:CBsuperlinear} and Lemma \ref{lem:algohyperlinear} together prove the first two claims of Theorem~\ref{thm:hex_query}. 
This shows that \GDCs have an asymptotic effect on query bounds and that we can take advantage of \GDCs to design new WCOJ algorithms to meet these bounds. Nevertheless, one might wonder whether the variable elimination idea of established WCOJ algorithms can already meet the improved bound and, in fact, achieve optimal runtimes. 
Proving the third claim in Theorem~\ref{thm:hex_query}, we show that they cannot and, thus, the new techniques have to be employed. 
Specifically, we show that VAAT algorithms (Definition~\ref{def:vaat}) require time $\Omega(n^{1.5})$ to compute $\hexq$ on database instances satisfying the PCs.

\begin{restatable}{lemma}{lemvaat}    
\label{lem:vaat}
    VAAT algorithms require time $\Omega(n^{1.5})$ to enumerate $\hexq^I$ for databases $I\vDash\mathbf{PC}$.
\end{restatable}

\begin{proof}[Proof Sketch.]
    It suffices to provide a collection of databases $\mathcal{I}$ such that $I\vDash \textbf{\GDC}$ and $\max_i|Q_{\varhexagon i}^I| = \Omega(n^{1.5})$ for databases $I\in \mathcal{I}$ and arbitrary ordering of the variables.
    To that end, we introduce two relations, a relation $C_{X,Y,Z}(X,Y,Z)$ and a set of disjoint paths $P_{X,Y,Z}(X,Y,Z)$.
    For $P_{X,Y,Z}(X,Y,Z)$, the domains of $X,Y,Z$ are of size $\Theta(n)$ and $P_{X,Y,Z}$ simply matches $X$ to $Y$ and $Z$ such that each $d\in dom(X)\cup dom(Y)\cup dom(Z)$ appears in exactly one tuple of $P_{X,Y,Z}$.
    On the other hand, think of $C_{X,Y,Z}$ as a complete bipartite graph from the domain of $X$ to the domain of $Y$ and $Z$ uniquely identifies the edges.
    Thus, $|dom(X)|=|dom(Y)| = \Theta(\sqrt{n})$ while $|dom(Z)| = \Theta(n)$.
    Notice that for both relations, $DC(XY,XYZ,1), DC(Z,XYZ,1)$ hold.
    I.e., any pair of variables determine the last variable and there is a variable that determines the whole tuple on its own.

    Consequently, the disjoint union of relations $C$ and $P$ multiple times (with permutated versions of $C$), e.g., 
    $$R(X,Y,Z) = C_{X,Y,Z}(X,Y,Z)\cupdot C_{Y,Z,X}(Y,Z,X) \cupdot P_{X,Y,Z}(X,Y,Z),$$ 
    satisfies $\GDC_{R}(\{X,Y,Z\}, XYZ, 1)$ and $DC(S_1S_2,XYZ,1)$ for any $S_1S_2 \subseteq XYZ$.

    Thus, we can set $R_1,R_2,R_3,R_4$ to be the disjoint union of $P$ and all permutations of the relation $C$.
    Now, let $X_1, \dots, X_6$ be an arbitrary variable order for $\hexq$.
    Then, a VAAT algorithm based on this variable order at least needs to compute the sets:
    \begin{alignat*}{5}
        &\Join_i\pi_{X_1}R_i, \quad&& \Join_i\pi_{X_1X_2}R_i, \quad &&\Join_i\pi_{X_1\cdots X_3}R_i, \\
        & \Join_i\pi_{X_1\cdots X_4}R_i, \quad &&\Join_i\pi_{X_1\cdots X_5}R_i, \quad &&\Join_i\pi_{X_1\cdots X_6}R_i.
    \end{alignat*}
    Let us consider the set $\Join_i\pi_{X_1\cdots X_4}R_i$.
    There are two cases:
    
    \textbf{Case 1:} $X_5$ and $X_6$ appear conjointly in a relation.
    Due to the symmetry of the query and the database, we can assume w.l.o.g, $X_5X_6 = UV$ and $\Join_i\pi_{X_1\cdots X_4}R_i = \Join_i\pi_{ABCW}R_i.$
    Furthermore, $\pi_{ABW}C_{A,B,W}\subseteq \pi_{ABW}R_1$, $\pi_{BC}C_{B,C,U}\subseteq \pi_{BC}R_2$, $ \pi_{AC}C_{A,C,V}\subseteq \pi_{AC}R_3, \pi_{W}P_{U,V,W}\subseteq \pi_{W}R_4$.
    Thus, intuitively, for at least $\Omega(\sqrt{n})$ elements $a\in dom(A)$ there are $\Omega(\sqrt{n})$ elements $b\in dom(B)$ and $\Omega(\sqrt{n})$ elements $c\in dom(C)$ that all join, and for each pair $a,b$ there is an element $w\in dom(W)$ that fits.
    In total, $|\Join_i\pi_{ABCW}R_i| = \Omega(n^{1.5})$.

    \textbf{Case 2:} $X_5$ and $X_6$ do not appear conjointly in a relation.
    Due to the symmetry of the query and the database, we can assume w.l.o.g, $X_5X_6 = AU$ and $\Join_i\pi_{X_1\cdots X_4}R_i $ $= \Join_i\pi_{BCVW}R_i.$
    Furthermore, $\pi_{BW}C_{B,W,A}\subseteq\pi_{BW}R_1$, $\pi_{BC}C_{B,C,U}\subseteq \pi_{BC}R_2$, $ \pi_{CV}C_{C,V,A}\subseteq \pi_{CV}R_3, \pi_{VW}C_{V,W,U}\subseteq \pi_{VW}R_4$.
    Thus, intuitively, for at least $\Omega(\sqrt{n})$ elements $b\in dom(B)$ there are $\Omega(\sqrt{n})$ elements $w\in dom(W)$, $\Omega(\sqrt{n})$ elements $v\in dom(V)$ and,  $\Omega(\sqrt{n})$ elements $c\in dom(C)$ that all join.
    In total, $|\Join_i\pi_{BCVW}R_i| = \Omega(n^{2})$.
\end{proof}

%% file: 5-general.tex
In the following, we extend the ideas of Section \ref{sec:hyprtriangle} to an arbitrary full conjunctive queries $Q(\mathbf{Z}) \leftarrow R_1(\mathbf{Z}_1)\Join \ldots\Join R_k(\mathbf{Z}_k)$ and an arbitrary set of partition constraints $\mathbf{\GDC} = \{\GDC_{P_1}(\mathcal{X}_1, {\bf Y}_1, d_1), \dots, \GDC_{P_l}(\mathcal{X}_l, {\bf Y}_l, d_l)\}$.
Recall, Algorithm \ref{alg:hypertriangle} proceeds by decomposing $R_4$ in accordance with the \GDC and then, executes a VAAT style WCOJ algorithm over the decomposed instances. 
We proceed in the same way and start by concentrating on decomposing relations. 
After partitioning the relations, we show how to lift DC-based techniques for bounding and enumerating conjunctive queries to PCs.

\subsection{Computing Constraints and Partitions}

To take advantage of \GDCs, we need to be able to decompose an arbitrary relation $R(\textbf{Z})$ according to a given partition constraint $\GDC(\mathcal{X},\textbf{Y}, d)$. For this task, we propose two poly-time algorithms; a linear approximate algorithm and a quadratic exact algorithm. Crucially, these algorithms do not need $d$ to be computed beforehand, so these algorithms can also be used to compute \GDC constraints themselves, i.e. to compute the value of $\GDC(\mathcal{X},\textbf{Y})$. We start with the faster approximation algorithm before describing the exact algorithm.

Concretely, Algorithm \ref{alg:approxdecomp} partitions a relation $R(\textbf{Z})$ by distributing the tuples from the relation to partitions in a greedy fashion. At each point, it selects the set of variables $\textbf{X}\in \mathcal{X}$ and particular value $\textbf{x}\in \pi_\textbf{X}R$ which occurs in the fewest tuples in the relation $R$. It then adds those tuples to the partition $R^{\mathbf{X}}$ and deletes them from the relation $R$. Intuitively, high degree pair ($\textbf{X}, \textbf{x}$) will be distributed to partitions late in this process. At this point, most of the matching tuples will already have been placed in different partitions.
Formally, we claim the following runtime and approximation guarantee for Algorithm \ref{alg:approxdecomp}.

\begin{algorithm}
    \begin{algorithmic}[1] 
        \State \textbf{decompose}$(R,\mathcal{X},\textbf{Y})$
        \For{$\textbf{X}\in \mathcal{X}$}
            \State $R^{\textbf{X}} \gets \emptyset$
        \EndFor
        \While{$R$ is not empty}
            \State $\textbf{X},\textbf{x} \gets \argmin_{\textbf{X}\in \mathcal{X}, \textbf{x}\in \pi_{\textbf{X}}R}|\pi_\textbf{Y}\sigma_{\textbf{X} = \textbf{x}}R|$ 
            \State $R^\textbf{X} \gets R^\textbf{X} \cup \sigma_{\textbf{X} = \textbf{x}}R$
            \State $R \gets R \setminus \sigma_{\textbf{X} = \textbf{x}}R$
        \EndWhile\\
        \Return $(R^\textbf{X})_{\textbf{X}\in \mathcal{X}}, \max_{\textbf{X}\in \mathcal{X}}DC_{R^\textbf{X}}(\textbf{X},\textbf{Y})$
    \end{algorithmic}
    \caption{Approximate Decomposition Algorithm}
    \label{alg:approxdecomp}
\end{algorithm}

\begin{restatable}{theorem}{thmalgoAprox}
\label{thm:algoAprox}
    For a relation $R(\mathbf{Z})$ and subsets $\mathbf{Y}\subseteq \mathbf{Z}, \mathcal{X}\subseteq 2^{\mathbf{Y}}$, Algorithm \ref{alg:approxdecomp} computes a partitioning $\bigcupdot_{\mathbf{X} \in \mathcal{X}} R^{\mathbf{X}} = R$ in time $O(|R|)$ (data complexity) such that 
        $$DC_{R^\mathbf{X}}(\mathbf{X}, \mathbf{Y}, |\mathcal{X}|d)$$
    holds for every $\mathbf{X}\in \mathcal{X}$ where $d=\GDC(\mathcal{X},\mathbf{Y})$.
\end{restatable}
\begin{proof}[Proof Sketch.]
    We start by providing intuition for the linear runtime. Each iteration of the while loop (line 4) places a set of tuples in a partition (line 6) and removes them from the original relation (line 7). Both of these operations can be done in constant time per tuple, so we simply need to show that we can identify the lowest degree attribute/value pair in constant time each iteration. This is done by creating a priority queue structure for each $\mathbf{X}\in\mathcal{X}$ where priority is equal to degree, and we begin by adding each tuple in $R$ to each priority queue. Because the maximum degree is less than $|R|$, we can construct these queues in linear time using bucket sort. We will then decrement these queues by 1 each time a tuple is removed from $R$. While arbitrarily changing an element's priority typically requires $O(log(|R|))$ in a priority queue, we are merely decrementing by 1 which is a local operation that can be done in constant time. So, construction and maintenance of these structures is linear w.r.t.\ data size. We can then use these queues to look up the lowest degree attribute/value pair in constant time.

    Next, we prove the approximation guarantee by contradiction. If the algorithm produces a partition $R^{\mathbf{X}}$ with $DC_{R^{\mathbf{X}}}(\mathbf{X},\mathbf{Y}) > |\mathcal{X}|d$, then there must be some value $\mathbf{x}\in\pi_\mathbf{X}R^{\mathbf{X}}$ where $|\pi_{\mathbf{Y}}\sigma_{\mathbf{X}=\mathbf{x}}R^{\mathbf{X}}| > |\mathcal{X}|d$. At the moment before this value was inserted into $R^{\mathbf{X}}$ and deleted from $R$, all attribute/value pairs must have had degree at least $|\mathcal{X}|d$ in the current state of $R$ which we denote $R_{\mathcal{A}}$. Through some algebraic manipulation, we show that this implies $|R_{\mathcal{A}}| > \sum_{\mathbf{X}\in\mathcal{X}}|\pi_{\mathbf{X}}R_{\mathcal{A}}|d$. On the other hand, we know that $R_{\mathcal{A}}$ respects the original partition constraint because $R_{\mathcal{A}}\subseteq R$, and we show that this implies the converse $ |R_{\mathcal{A}}| \leq \sum_{\mathbf{X}\in\mathcal{X}}|\pi_{\mathbf{X}}R_{\mathcal{A}}|d$.
    This is a contradiction, so our algorithm must not produce a poor approximation in the first place.
\end{proof}

Next, we describe an exact algorithm that requires quadratic time. Intuitively, Algorithm~\ref{alg:exactdecomp} also computes a decomposition of $R$ in a greedy fashion by iteratively allocating  (groups of) tuples $\textbf{y}_0$ of $\pi_Y R$ to partitions $R^\textbf{X}$, preferring allocations to partitions where the maximum over the relevant degree constraints, i.e., $\max_\textbf{X} DC_{R^\textbf{X}}(\textbf{X},\textbf{Y})$, does not increase.  However, decisions greedily made at the start may be sub-optimal and may not lead to a decomposition that minimizes $\max_\textbf{X} DC_{R^\textbf{X}}(\textbf{X},\textbf{Y})$. To overcome this, instead of simply allocating $\textbf{y}_0$ to a partition, the algorithm also checks whether it is possible to achieve a better overall decomposition by reallocating some other tuples in a cascading manner. To that end, we look for elements $\textbf{y}_1 \in \pi_Y R^{\textbf{X}_1} ,\dots,\textbf{y}_m\in \pi_Y R^{\textbf{X}_m}$ and a further $\textbf{X}_{m+1}$ such that for all $\textbf{y}_i, i=1,\dots,m$, the tuples matching $\textbf{y}_i$ can be moved from $R^{\textbf{X}_i}$ to $R^{\textbf{X}_{i+1}}$ and the tuples matching $\textbf{y}_0$ can be added to $R^{\textbf{X}_1}$, all without increasing $\max_\textbf{X} DC_{R^\textbf{X}}(\textbf{X},\textbf{Y})$. To achieve this, $\textbf{y}_{i}$ is selected such that it matches $\textbf{y}_{i-1}$ on the variables $\textbf{X}_{i}$. Thus, for each $R^{\textbf{X}_{i}}$, the value of $DC_{R^{\textbf{X}_{i}}}(\textbf{X},\textbf{Y})$ is the same before and after the update. There only has to be space for $\textbf{y}_m$ in the final relation $R^{\textbf{X}_{m+1}}$.

The sequence $(\textbf{y}_1,\dots,\textbf{y}_m,\textbf{X}_1,\dots,\textbf{X}_{m+1})$ constitutes an augmenting path which was first introduced by \cite{edmonds1965minimum} and used for matroids.
Adapted to the present setting, we define an augmenting path as below.
By slight abuse of notation, we write $\sigma_{\textbf{X}=\textbf{y}}R$ even though $\textbf{y}$ fixes more variables than specified in $\textbf{X}$.
Naturally, this is meant to select those tuples of $R$ that agree with $\textbf{y}$ on the variables $\textbf{X}$.

\begin{definition}
\label{def:augmenting}
    Let $(R^\mathbf{X})_{\mathbf{X}\in \mathcal{X}}$ be pairwise disjoint subsets of some relation $R(\mathbf{Z})$ with  $\mathbf{Y}\subseteq \mathbf{Z}, \mathcal{X} \subseteq 2^{\mathbf{Y}}$, and let $\mathbf{y}_0 \in \pi_{\mathbf{Y}}R\setminus \pi_\mathbf{Y}\bigcupdot_{\mathbf{X}\in\mathcal{X}} R^\mathbf{X}$ be a new tuple.
    An augmenting path $(\mathbf{y}_1,\dots,\mathbf{y}_m,\mathbf{X}_1,\dots,\mathbf{X}_{m+1})$ satisfies the following properties:
\begin{enumerate}
    \item For all $i\in \{1,\dots, m+1\}: \mathbf{X}_i\in \mathcal{X}$.
    \item For all $i\in \{1,\dots, m\}:\mathbf{y}_i \in \pi_\mathbf{Y}R^{\mathbf{X}_i}$.
    \item For all $i\in \{1,\dots, m\}:\mathbf{y}_i $ agrees with $\mathbf{y}_{i-1}$ on $\mathbf{X}_{i}$.
    \item \label{prop:dcleq} $|\pi_\mathbf{Y}\sigma_{\mathbf{X}_{m+1} = \mathbf{y}_m}R^{\mathbf{X}_{m+1}}| < \max_\mathbf{X} DC_{R^\mathbf{X}}(\mathbf{X},\mathbf{Y})$.
\end{enumerate}
We omit the references to $(R^\mathbf{X})_{\mathbf{X}\in \mathcal{X}}$, $\mathbf{Y}$, and $\mathbf{y}_0$ when they are clear from the context.
\end{definition}

The next theorem shows that Algorithm \ref{alg:exactdecomp} correctly computes an optimal decomposition.

\begin{algorithm}
    \begin{algorithmic}[1] 
        \State \textbf{decompose}$(R,\mathcal{X},\textbf{Y})$
        \State $d \gets 0$
        \For{$\textbf{X}\in \mathcal{X}$}
            \State $R^{\textbf{X}} \gets \emptyset$
        \EndFor
        \For{$\textbf{y}_0 \in \pi_\textbf{Y}R$}
            \If{there exists a shortest augmenting path $(\textbf{y}_1,\dots,\textbf{y}_m,\textbf{X}_1,\dots,\textbf{X}_{m+1})$ }
                \For{$i=m,\dots,1$}
                    \State $R^{\textbf{X}_{i+1}} \gets R^{\textbf{X}_{i+1}}\cup \sigma_{\textbf{Y}=\textbf{y}_i}R^{\textbf{X}_i}$
                    \State $R^{\textbf{X}_{i}} \gets R^{\textbf{X}_{i}}\setminus \sigma_{\textbf{Y}=\textbf{y}_i}R^{\textbf{X}_i}$
                \EndFor
                \State $R^{\textbf{X}_1} \gets R^{\textbf{X}_1} \cup \sigma_{\textbf{Y} = \textbf{y}_0}R$
            \Else
                \State $d\gets d+1$
                \State $R^\textbf{X} \gets R^\textbf{X} \cup \sigma_{\textbf{Y} = \textbf{y}_0}R$ \Comment{$\textbf{X}\in \mathcal{X}$ can be selected arbitrarily here.}
            \EndIf
                \State $R \gets R \setminus \sigma_{\textbf{Y} = \textbf{y}_0}R$
        \EndFor\\
        \Return $(R^\textbf{X})_{\textbf{X}\in \mathcal{X}}, d$
    \end{algorithmic}
    \caption{Exact Decomposition Algorithm}
    \label{alg:exactdecomp}
\end{algorithm}

\begin{restatable}{theorem}{thmalgoExact}
\label{thm:algoExact}
    For a relation $R(\mathbf{Z})$ and subsets $\mathbf{Y}\subseteq \mathbf{Z}, \mathcal{X}\subseteq 2^{\mathbf{Y}}$, Algorithm \ref{alg:exactdecomp} computes a partitioning $\bigcupdot_{\mathbf{X} \in \mathcal{X}} R^{\mathbf{X}} = R$ in $O(|R|^2)$ time (data complexity) such that 
        $$DC_{R^\mathbf{X}}(\mathbf{X}, \mathbf{Y}, d)$$
    holds for every $\mathbf{X}\in \mathcal{X}$ where $d=\GDC(\mathcal{X},\mathbf{Y})$.
\end{restatable}

\begin{proof}[Proof Sketch.]
The intuition for the algorithm's quadratic runtime boils down to ensuring that the search for an augmenting path is linear in $|R|$. To show that this is the case, we model the search for an augmenting path as a breadth-first search over a bipartite graph whose nodes correspond to full tuples $\mathbf{y}\in\pi_{\mathbf{Y}}R$ and the partial tuples $\mathbf{x}\in \pi_{\mathbf{X}}R$ for all $\mathbf{X}\in\mathcal{X}$. An edge exists between a tuple $\mathbf{y}$ and a partial tuple $\mathbf{x}$ if they agree on their shared attributes. This search starts from the new tuple $\mathbf{y}_0$ and completes when it finds a tuple $\mathbf{y}_{m}$ that can be placed in one of the relations $R^{\mathbf{X}}$ without increasing its DC. The number of edges and vertices in this graph is linear in the input data, so the breadth-first search is linear as well.

We now provide intuition for the algorithm's correctness. Suppose that we are partway through the algorithm and are now at the iteration where we add $\mathbf{y}_0$ to a partition. Further, let $R_{\mathcal{A}}$ be the set of tuples which have been added so far, including $\mathbf{y}_0$. Because there exists an optimal partitioning of $R_{\mathcal{A}}$, we should be able to place $\mathbf{y}_0$ in the partition where it exists in the optimal partitioning. If we cannot, then a tuple in that partition with the same $\textbf{X}$-value must not be in its optimal partition. We identify one of these tuples and move it to its optimal partition. We continue this process inductively of shifting one tuple at a time to its optimal partition until one of these shifts no longer violates the $PC$. The sequence of shifts that we have performed constitutes an augmenting path by definition, and its existence implies that we would not violate the $PC$ in this iteration by incrementing $d$ past it. This construction process must end in a finite number of moves because each step increases the number of tuples placed in their optimal partitioning. If all tuples are in their optimal partition, the final shift must not have violated the $PC$ due to the optimal partition's definition.
\end{proof}

\subsection{Lifting DC-Based Bounds and Algorithms to PCs}

We now apply the developed decomposition algorithms to the general case and show how to use WCOJ algorithms as a black box to carry over guarantees for DCs to the general case of \GDCs.
First, we extend the definition of an arbitrary cardinality bound $CB$ defined over sets of DCs to sets of PCs. 
\begin{definition}
    \label{def:GCB}
    Let $CB$ be a cardinality bound. Then, we extend $CB$ to \GDCs by setting
    \begin{align*} 
    CB(Q,\mathbf{\GDC}) &:= \sum_{\mathbf{X}^1\in  \mathcal{X}_1,\cdots, \mathbf{X}^l\in  \mathcal{X}_l}
    CB(Q,\{DC_{P_1}(\mathbf{X}^1, {\bf Y}_1, d_1), \dots, DC_{P_l}(\mathbf{X}^l, {\bf Y}_l, d_l)\}), 
    \end{align*}
    where $\mathbf{\GDC} = \{\GDC_{P_i}(\mathcal{X}_i, {\bf Y}_i, d_i)\mid i\}$ is an arbitrary set of partition constraints.
\end{definition}

Note that the bound in Definition \ref{def:GCB} is well-defined: If $\mathbf{\GDC}$ is simply a set of DCs, the extended version of the cardinality bound $CB$ coincides with its original definition. If $\mathbf{\GDC}$ is an arbitrary set of $\GDCs$, then it remains a valid bound on the size of the join.

\begin{restatable}{proposition}{propgdc}
    \label{prop:gcb}
    Let $CB(Q,\mathbf{\GDC})$ be an extended cardinality bound. 
    Then,
    $$|Q^I| \leq CB(Q,\mathbf{\GDC}) \quad \forall I\vDash \mathbf{\GDC}.$$
\end{restatable}

\begin{proof}
    Let $I$ be a database satisfying $\textbf{\GDC} = \{\GDC_{P_i}(\mathcal{X}_i, {\bf Y}_i, d_i)\mid i=1,\dots,l\}$ and $Q\leftarrow R_1(\textbf{Z}_1)\Join \dots \Join R_k(\textbf{Z}_k)$.    
    Thus, for each $P_i$ there is a partitioning $P_i= \bigcupdot_{j=1}^{|\mathcal{X}_i|}P_{i}^{j}$ with $DC_{P_{i}^{j}}(\textbf{X}_{i}^{j}, \textbf{Y}_i, d_i)$ and $\mathcal{X}_i = \{X_{i}^{j}| j=1,\dots, |\mathcal{X}_i|\}$.
    Now, let us now fix some $j_1 \in \{1,\dots,$ $ |\mathcal{X}_1|\}, \dots, $ $ j_l \in \{1,\dots, |\mathcal{X}_l|\}$.
    Furthermore, for each $i\in \{1,\dots, l\}$ let $s(i)\in \{1,\dots, k\}$ be the relation $R_{s(i)}=P_i$.
    Set $R_{s}^{j_1\dots j_l} = \bigcap_{i:s=s(i)}P_{i}^{j_i}$ for all $s\in \{1,\dots, k\}$.
    Then consider $Q^{j_1\dots j_l} = R_{1}^{j_1\dots j_l}\Join \dots \Join R_{k}^{j_1\dots j_l}$.
    We claim two things:
    \begin{enumerate}
        \item $Q^{j_1\dots j_l}$ partitions $Q^I$, i.e., $Q^I = \bigcup_{j_1=1}^{|\mathcal{X}_1|}\dots \bigcup_{j_l=1}^{|\mathcal{X}_l|}Q^{j_1\dots j_l}$, and
        \item $|Q^{j_1\dots j_l}| \leq CB(Q,\{DC_{P_1}(\textbf{X}_{1}^{j_1}, {\bf Y}_1, d_1), \dots, DC_{P_l}(\textbf{X}_{l}^{j_l}, {\bf Y}_l, d_l)\})$.
    \end{enumerate}
    For the first bullet point, let $t\in Q^I$.
    Clearly, for each $i\in \{1,\dots, l\}$ there exists exactly one $j_i$ such that $t$ agrees with an element in $P_{i}^{j_i}(\textbf{Z}_{s(i)})$ on the variables $\textbf{Z}_{s(i)}$.
    Thus, $t\in Q^{j_1\dots j_l}$.

    For the second bullet point, consider the relations $P_{i}^{j_i}$.
    These are supersets of the relations $R_{s(i)}^{j_1\dots j_l}$ and, therefore, we can assert $DC_{R_{s(i)}^{j_1\dots j_l}}(\textbf{X}_{i}^{j_i}, \textbf{Y}_i, d_i)$.
    Viewed as a query, $Q^{j_1\dots j_l}$ has the same form as $Q$.
    Hence, $|Q^{j_1\dots j_l}| \leq CB(Q,\{DC_{P_1}(\textbf{X}_{1}^{j_1}, {\bf Y}_1, d_1), \dots, DC_{P_l}(\textbf{X}_{l}^{j_l}, {\bf Y}_l, d_l)\})$.
\end{proof}

Crucially, this upper bound is not loose; it preserves the tightness of any underlying DC-based bound. Specifically, the extended version of the combinatorics bound $CB_{Comb}$ is asymptotically close to the actual worst-case size of the join, with the constant depending on the query. For example, the extended version of combinatorics bound is $O(n)$ for the query $\hexq$ (see Section \ref{sec:hyprtriangle}) when using all \GDCs while the combinatorics bound based on the DCs alone was $\Omega(n^{\frac{4}{3}})$.


\begin{restatable}{proposition}{propcomb}
\label{prop:combPC}
    Let $CB_{Comb}(Q,\mathbf{\GDC})$ be the extended version of the combinatorics bound $CB_{Comb}$. 
    Then,
    $$CB_{Comb}(Q,\mathbf{\GDC}) = O(\max_{I\vDash \mathbf{\GDC}}|Q^I|).$$
\end{restatable}
\begin{proof}
    Let $\textbf{X}^1\in  \mathcal{X}_1, \dots, \textbf{X}^l\in  \mathcal{X}_l$ be such that $CB_{Comb}(Q,\textbf{DC})$ is maximized where $\textbf{DC}:= \{DC_{P_1}(\textbf{X}^1,$ $ {\bf Y}_1, d_1), \dots,$ $ DC_{P_l}(\textbf{X}^l, {\bf Y}_l, d_l)\}$.
    Thus, there exists a database $I$ such that $I\vDash \textbf{DC}$ and $|Q^I|=CB_{Comb}(Q,\textbf{DC})$.
    Furthermore, $I$ must then also trivially satisfy $\textbf{\GDC}$ by definition.
    For each \GDC on $P_i$ the witnessing partitioning is simply $P_i = P_i \cupdot \emptyset \cupdot \dots \cupdot \emptyset$.
    Therefore, 
    $$CB_{Comb}(Q,\mathbf{\GDC}) \leq |\mathcal{X}_1| \dots |\mathcal{X}_l| |Q^I| = O(\max_{D\vDash \mathbf{\GDC}}|Q^I|) .$$
    For the last equality, note that $|\mathcal{X}_1| \dots |\mathcal{X}_l|$ is a query-dependent constant as we assume all \GDCs to be on different variables $\mathcal{X}_i$.
\end{proof}

Next, we show how algorithms that are worst-case optimal relative to a cardinality bound can likewise be adapted and become worst-case optimal relative to the extended bound. In this way, progress on WCOJ algorithms based on DCs immediately leads to improved algorithms that take advantage of PCs.


\begin{restatable}{theorem}{thmgeneral}
    Given a cardinality bound $CB_{\mathcal{A}}$. 
    If there exists a join enumeration algorithm $\mathcal{A}$ which is worst-case optimal relative to $CB_{\mathcal{A}}$, then there exists an algorithm $\mathcal{A}^*$ which is worst-case optimal relative to the extended version of the cardinality bound. 
    I.e., for fixed query $Q$, $\mathcal{A}^*$ runs in time $O(|I|+CB_{\mathcal{A}}(Q,\mathbf{\GDC}))$ for arbitrary $\mathbf{\GDC}$ and database $I\vDash \mathbf{\GDC}$.
\end{restatable}
\begin{proof}
    Suppose we have a query $Q\leftarrow R_1(\textbf{Z}_1)\Join \dots R_k(\textbf{Z}_k)$, a set of partition constraints $\mathbf{\GDC} = \{\GDC_{P_1}(\mathcal{X}_1, {\bf Y}_1, d_1), \dots, \GDC_{P_l}(\mathcal{X}_l, {\bf Y}_l, d_l)\}$, and a database $I \vDash \textbf{\GDC}$. 
    We can follow the same idea already used in the proof of Proposition \ref{prop:gcb}.
    Concretely, we partition each $P_i$ into $(P_i^j)_{j=1,\dots,|\mathcal{X}_i|}$.
    For this, we use Algorithm \ref{alg:approxdecomp} and Theorem \ref{thm:algoAprox}.
    Thus, we get $DC_{P_{i}^{j}}(\textbf{X}_{i}^{j}, \textbf{Y}_i, |\mathcal{X}_i|d_i)$. 
    Let $\alpha:= \max_i|\mathcal{X}_i|$.

    We can now continue following the idea of proof of Proposition \ref{prop:gcb} up to the two claims
    where we only need the first.
    Now, instead of bounding the size of $Q^{j_1,\dots,j_l}$, we now want to compute each subquery using $\mathcal{A}$.
    Thus, note that $DC_{R_{s(i)}^{j_1\dots j_l}}(\textbf{X}_{i}^{j_i}, \textbf{Y}_i, \alpha d_i)$.
    This implies that $\mathcal{A}$ runs in time $O(|I|+CB_{\mathcal{A}}(Q,\{DC_{P_i}(\textbf{X}_i^{j_i}, \textbf{Y}_i, \alpha d_i)\mid i\})).$
    The constant $\alpha$ can be hidden in the $O$-notation while summing up the time required for each $Q^{j_1,\dots,j_l}$ results in an overall runtime of 
    $O(|I|+\sum_{\textbf{X}^1\in  \mathcal{X}_1\cdots\textbf{X}^l\in  \mathcal{X}_l}CB_{\mathcal{A}}(Q,\{DC_{P_i}(\textbf{X}_i^{j_i}, \textbf{Y}_i, d_i)\mid i\}) = O(|I|+CB_{\mathcal{A}}(Q,\textbf{\GDC})).$
\end{proof}

For example, PANDA is a WCOJ algorithm relative to the polymatroid bound, and we can use this approach to translate it to an optimal algorithm for the extended polymatroid bound. While it is an open problem to produce a WCOJ algorithm relative to $CB_{Comb}(Q,\mathbf{DC})$, any such algorithm now immediately results in a WCOJ algorithm relative to  $CB_{Comb}(Q,\mathbf{\GDC})$. Combined with Proposition~\ref{prop:combPC} this implies that such an algorithm then only takes time relative to the worst-case join size of instances that satisfy the same set of PCs.

\begin{corollary}
    If there exists a WCOJ algorithm relative to $CB_{Comb}(Q,\mathbf{DC})$, then there exists an WCOJ algorithm relative to $CB_{Comb}(Q,\mathbf{\GDC})$.
\end{corollary}

%% file: 7-conclusion.tex
In this work, we introduced PCs as a generalization of DCs, uncovering a latent structure within relations and that is present in standard benchmarks. PCs enable a more refined approach to query processing, offering asymptotic improvements to both cardinality bounds and join algorithms. We presented algorithms to compute PCs and identify the corresponding partitioning that witness these constraints. To harness this structure, we then developed techniques to lift both cardinality bounds and WCOJ algorithms from the DC framework to the PC framework. Crucially, our use of DC-based bounds and algorithms as black boxes allows future advances in the DC setting to be seamlessly integrated into the PC framework.

On the practical side, future research should explore when and where it is beneficial to leverage the additional structure provided by PCs. In particular, finding ways to minimize the constant factor overhead by only considering a useful subset of PCs or sharing work across evaluations on different partitions could yield significant practical improvements in query performance. On the other hand, further theoretical work should try to incorporate additional statistics into this partitioning framework, e.g., $l_p$-norms of degree sequences. Ultimately, the goal of this line of work is to capture the inherent complexity of join instances through both the structure of the query and the data.

%% file: app-4.tex
\section{Further Details for Section \ref{sec:gdc}}

\thmpre*
\begin{proof}
To prove this theorem, we assume towards a contradiction that there exists a partitioning $(R^1,\dots,R^k)$ of $R$ such that,
    \begin{align*}
        \max_{i=1,\dots,k}\GDC_{R^i}(\mathcal{X},\mathbf{Y})< \GDC_{R}(\mathcal{X},\mathbf{Y})/k.
    \end{align*}
Then, we can partition each $R^i$ such as to witness $\GDC_{R^i}(\mathcal{X},\mathbf{Y})$. 
I.e., let $\bigcup_{\textbf{X}\in \mathcal{X}} R^{i,\textbf{X}}=R^i$ be such that $\max_{\textbf{X}\in \mathcal{X}}DC_{R^{i,\textbf{X}}}(\textbf{X}, \textbf{Y}) = \GDC_{R^i}(\mathcal{X},\mathbf{Y})$.

Then, for each fixed $\textbf{X}\in \mathcal{X}$, we can combine the subrelations $R^{1,\textbf{X}}, \dots, R^{k,\textbf{X}}$ into a single relation $R^\textbf{X}:=\bigcup_{i=1}^k R^{i,\textbf{X}}$.
These relations $(R^\textbf{X})_{\textbf{X}\in \mathcal{X}}$ then together partition $R$ as
\[\bigcup_{\textbf{X}\in \mathcal{X}}R^\textbf{X} = \bigcup_{\textbf{X}\in \mathcal{X}}\bigcup_{i=1}^k R^{i,\textbf{X}} = \bigcup_{i=1}^k R^{i} = R.\]
Now, let us compute the degree constraints for each $R^\textbf{X}$.
To that end, let $\textbf{x}\in \pi_{\textbf{X}} R^\textbf{X}$ and $R^{\textbf{X},\textbf{x}} = \sigma_{\textbf{X}=\textbf{x}}R^\textbf{X}$.
We can partition $R^{\textbf{X},\textbf{x}}$ according to $(R^{i,\textbf{X}})_i$ and set $R^{i,\textbf{X},\textbf{x}}:=R^{\textbf{X},\textbf{x}} \cap R^{i,\textbf{X}}$.
Due to the degree constraint $DC_{R^{i,\textbf{X}}}(\textbf{X}, \textbf{Y})$ we know that $|\pi_\textbf{Y}R^{i,\textbf{X},\textbf{x}}| \leq \GDC_{R^i}(\mathcal{X},\mathbf{Y})$ for every $i=1,\dots,k$.
Thus, in total 
\[|\pi_{\textbf{Y}}R^{\textbf{X},\textbf{x}}| = |\pi_{\textbf{Y}}(\bigcup_i R^{i,\textbf{X},\textbf{x}})|\leq \sum_i |\pi_{\textbf{Y}}(R^{i\textbf{X},\textbf{x}})|\leq k\cdot\GDC_{R^i}(\mathcal{X},\mathbf{Y}) < \GDC_{R}(\mathcal{X},\mathbf{Y}).\]
Note, that the last inequality holds due to our assumption.
As $\textbf{x}$ was selected arbitrarily we get $DC_{R^\textbf{X}}(\textbf{X},\textbf{Y})<\GDC_{R}(\mathcal{X},\mathbf{Y})$.
Furthermore, as $\textbf{X}$ was selected arbitrarily and $(R^\textbf{X})_{\textbf{X}\in \mathcal{X}}$ partitions $R$, we get that 
\[PC_R(\mathcal{X},\textbf{Y}) \leq \max_{\textbf{X}\in \mathcal{X}}DC_{R^\textbf{X}}(\textbf{X},\textbf{Y}) < PC_R(\mathcal{X},\textbf{Y})\]
due to the definition of partition constraints.
\end{proof}

\section{Further Details for Section \ref{sec:hyprtriangle}}
\label{app:hyprtriangle}

\lemCBsuperlinear*
\begin{proof}
    It suffices to provide a collection of databases $\mathcal{I}$ such that the $I\vDash \textbf{DC}$ and $|\hexq^I| = \Omega(n^\frac{4}{3})$ for $I\in \mathcal{D}$.
    To accomplish this, we introduce a new relation $R$ defined by (the subscripts $X,Y,Z$ are only used for clarity as to what constant belongs to the domain of which attribute)
    $$R_{X,Y,Z} = \{(i_X, j_Y, (i + j - \lfloor \frac{1}{2}n^\frac{1}{3} \rfloor \mod n^\frac{2}{3})_Z) \mid i = 0,\dots, n^\frac{2}{3}-1, j = 0,\dots, n^\frac{1}{3}-1\}.$$
    For simplicity, we assume $n$ to be the cube of an odd number.

    Intuitively, think of $R$ as a bipartite graph from the domain of $X$ to the domain of $Z$ where $Y$ identifies the edge for a given $x\in dom(X)$ or $z\in dom(Z)$.
    The domain of $X$ and $Z$ is of size $n^{\frac{2}{3}}$ while the domain of $Y$ is of size $n^{\frac{1}{3}}$.
    Thus, every $x\in dom(X)$ has $n^{\frac{1}{3}}$ neighbors in $dom(Z)$ and, due to symmetry, also the other way around.
    $R_{X,Y,Z}$ is depicted in Figure \ref{fig:mod} for $n=27$.
    \begin{figure}
        \centering
        \includegraphics[width=0.6\textwidth]{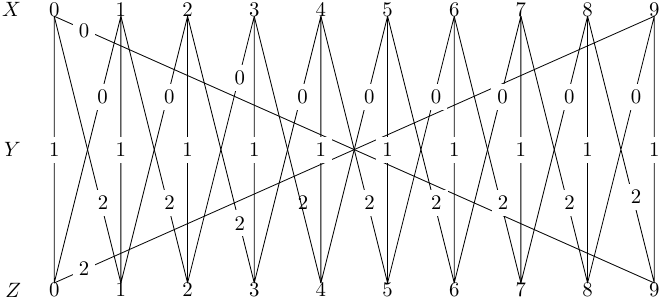}
        \caption{Depiction of $R_{X,Y,Z}$ for $n=27$}
        \label{fig:mod}
    \end{figure}
    
    Notice that, $DC_{R_{X,Y,Z}}(XY,XYZ,1)$ and $DC_{R_{X,Y,Z}}(YZ,XYZ,1)$.
    We set
    \begin{align*}
        R_1 = R_{A,W,B}, \quad\quad R_2 = R_{B,U,C}, \quad\quad R_3 = R_{C,V,A}.
    \end{align*}
    Lastly, we set 
    $$R_4 = \{(i_{U}, j_{V}, k_{W}) \mid i,j,k = 1,\dots, n^\frac{1}{3}\}.$$

    Clearly, the database constructed in this way satisfies \textbf{DC} (but not \textbf{PC}).
    Now let us compute the number of answers in $\hexq^I$.
    First notice that $R_4 = dom(U) \times dom(V) \times dom(W)$.
    Thus, we can simply ignore $R_4$ and it suffices to look at $\pi_{AB}R_1 \Join \pi_{BC}R_2, \Join \pi_{AC}R_3$.
    
    Let $i_A\in dom(A)$ be arbitrary.
    It is connected to $\{(i-\lfloor\frac{1}{2}n^\frac{1}{3}\rfloor)_B, \dots, (i + \lfloor\frac{1}{2}n^\frac{1}{3}\rfloor)_B\}$ via $\pi_{AB}R_1$ and to $\{(i-\lfloor\frac{1}{2}n^\frac{1}{3}\rfloor)_C, \dots, (i + \lfloor\frac{1}{2}n^\frac{1}{3}\rfloor)_C\}$ via $\pi_{AC}R_3$.
    Furthermore, $N(i,B) := \{(i-\lfloor\frac{1}{4}n^\frac{1}{3}\rfloor)_B, \dots, (i + \lfloor\frac{1}{4}n^\frac{1}{3}\rfloor)_B\}$ are all connected to $N(i,C):=\{(i-\lfloor\frac{1}{4}n^\frac{1}{3}\rfloor)_C, \dots, (i + \lfloor\frac{1}{4}n^\frac{1}{3}\rfloor)_C\}$ in $\pi_{BC}R_2$.
    Thus, 
    $$\{\{i_A\}\times N(i,B) \times N(i,C) \mid i = 0, \dots, n^\frac{2}{3}-1\}\subseteq \pi_{AB}R_1 \Join \pi_{BC}R_2, \Join \pi_{AC}R_3.$$
    
    Consequently, $|\hexq^I| = |\pi_{AB}R_1 \Join \pi_{BC}R_2, \Join \pi_{AC}R_3| = \Omega(n^\frac{4}{3})$.
\end{proof}

\lemalgohyperlinear*

\begin{proof}
    We start by proving that a partitioning $R_4 = R_{4}^{U} \cup R_{4}^{V} \cup R_{4}^{W}$ with $DC_{R_{4}^{U}}(U,UVW, 3),$ $ DC_{R_{4}^{V}}(V,UVW, 3), DC_{R_{4}^{W}}(W,UVW, 3)$ can be computed in linear time when $PC_{R_4}(\{U,V,W\},$ $ UVW, 1)$.
    Let us assume, w.l.o.g, that the domains of $U, V, W$ are disjoint and let $\mathcal{V}(d)$ be the variable associated with the constant $d\in \mathcal{D} :=dom(U)\cup dom(V)\cup dom(W)$. We claim that there exists a  $d\in \mathcal{D}$ such that $|\sigma_{\mathcal{V}(d) = d}R_4|\leq 3$, i.e., $d$ appears at most 3 times.

    Assume towards a contradiction that $|\sigma_{\mathcal{V}(d) = d}R_4| > 3$ for all $d\in \mathcal{D}$. Then, $|R_4|>3|dom(U)|, |R_4|>3|dom(V)|, |R_4|>3|dom(W)|$. However, $\GDC_{R_4}(\{U,V,W\}, UVW, 1)$ ensures that there is a three-way partition $R_4 = R_{4}^{U} \cup R_{4}^{V} \cup R_{4}^{W}$ satisfying $DC_{R_{4}^{U}}(U,UVW,1),$ $ DC_{R_{4}^{V}}(V,UVW,1),$ $ DC_{R_{4}^{W}}(W,UVW,1)$. 
    One of the parts, w.l.o.g., say $R_{4}^{U}$, then has to contain at least $\frac{1}{3}|R_4|$ tuples.
    Thus, $|R_4^{U}| \geq \frac{1}{3}|R_4| > |dom(U)|$.
    But, at the same time, $|R_{4}^{U}| = \sum_{d\in dom(U)}|\sigma_{U=d}R_{4}^{U}| \leq \sum_{d\in dom(U)}1 = |dom(U)|$ leads to a contradiction.

    Therefore, we can construct $R_{4}^{U}, R_{4}^{V}, R_{4}^{W}$ as follows: Initialize $R_{4}^{U}, R_{4}^{V}, R_{4}^{W}$ to be empty. Then, select a $d_1\in \mathcal{D}$ such that $\sigma_{\mathcal{V}(d_1) = d_1}R_4 \leq 3$, add $\sigma_{\mathcal{V}(d_1) = d_1}R_4$ to $R_{4}^{\mathcal{V}(d_1)}$, and remove $\sigma_{\mathcal{V}(d_1) = d_1}R_4$ from $R_4$. We can repeat this and select a $d_2\in \mathcal{D}$ such that $\sigma_{\mathcal{V}(d_2) = d_2}R_4 \leq 3$, add $\sigma_{\mathcal{V}(d_2) = d_2}R_4$ to $R_{4}^{\mathcal{V}(d_2)}$, and remove $\sigma_{\mathcal{V}(d_2) = d_2}R_4$ from $R_4$, and so on.

   Since each $d_i$ is unique and $\GDC_{R'_4}(\{U,V,W\}, UVW, 1)$ also holds for any $R'_4\subseteq R_4$, we can be sure that $DC_{R^X_{4}}(X,UVW, 3)$ holds for every $X=U,V,W$ along the way and in the end. 
   This process can be implemented to run in linear time by maintaining three priority queues, one for each variable $U,V,W$.
   To that end, we start by bucket sorting all tuples three times to get the values of $|\sigma_{\mathcal{V}(d) = d}R_4|$ for each $d\in \mathcal{D}$.
   These collections $\sigma_{\mathcal{V}(d) = d}R_4$ are then added to the queue for the variable $\mathcal{V}(d)$ with the key $|\sigma_{\mathcal{V}(d) = d}R_4|$.
   Lookups in these queue only happen up to a key value of 3 and thus only take constant time.
   For updates, when $\sigma_{\mathcal{V}(d_i) = d_i}R_4$ are added to a partition $R_{4}^{U}, R_{4}^{V}$ or $ R_{4}^{W}$, we simply have to update the values for the other constants in $\sigma_{\mathcal{V}(d_i) = d_i}R_4$.
   Concretely, lets assume $\mathcal{V}(d_i)=U$.
   Then, for each tuple $(u,v,w)$, we have to remove $(u,v,w)$ from $\sigma_{V = v}R_4$ and $\sigma_{W = w}R_4$ from the queue for $V$ and $W$, respectively, as well as updating the corresponding key values $|\sigma_{V = v}R_4|$ and $\sigma_{W = w}R_4$.
   The order in the queue is only important for key values up to a value of 3.
   Thus, updates can implemented in constant time.

    Now let us proceed to the for loops.
    Notice that the sizes of the projections are constant due to the DCs.
    For example, consider the first set of nested loops.
    Given a $(a,u,b)$ from $R_1$, there are at most 3 matching $(v,s)$ from $R_{4,U}$ due to $DC_{R_{4,U}}(U,UVS,3)$.
    Furthermore, there is at most one matching $c$ from $R_2$ and $R_3$, independanty due to $DC_{R_{2}}(BV,BVC,3)$ and $DC_{R_{3}}(AS,ASC,1)$.
    Thus, for the loops to only take linear time, the relations used in the nested loops simply need to be sorted such that the selection part only requires constant time.
    This is possible by hashing.

    The correctness of the algorithm follows since the sets of nested loops respectively compute $R_1\Join R_2\Join R_3\Join R_{4}^{U}$, $R_1\Join R_2\Join R_3\Join R_{4}^{V}$, and $R_1\Join R_2\Join R_3\Join R_{4}^{S}$.
    Together they equal $R_1\Join R_2\Join R_3\Join R_{4} = \hexq^I$ as required.
\end{proof}

\lemvaat*
\begin{proof}
    It suffices to provide a collection of databases $\mathcal{I}$ such that $I\vDash \textbf{\GDC}$ and $\max_i|Q_{\varhexagon i}^I| = \Omega(n^{1.5})$ for every $I\in \mathcal{I}$ and ordering of of the variables.
    To that end, we introduce two relations, a relation $C_{X,Y,Z}(X,Y,Z)$ and a set of disjoint paths $P_{X,Y,Z}(X,Y,Z)$:
    \begin{align*}
        C_{X,Y,Z} &= \{(i_{X},j_{Y},(i\sqrt{\frac{n}{7}}+j)_{Z}) \mid i,j= 0,\dots, \sqrt{\frac{n}{7}}-1\}, \\
        P_{XYZ} &= \{(i_{X},i_{Y},i_{Z}) \mid i= 0,\dots, \frac{n}{7}-1\}.
    \end{align*}
    For $P_{X,Y,Z}(X,Y,Z)$, the domains of $X,Y,Z$ are of size $\Theta(n)$ and $P_{X,Y,Z}$ simply matches $X$ to $Y$ and $Z$ such that each $d\in dom(X)\cup dom(Y)\cup dom(Z)$ appears in exactly one tuple of $P_{X,Y,Z}$.
    On the other hand, think of $C_{X,Y,Z}$ as a complete bipartite graph from the domain of $X$ to the domain of $Y$ and $Z$ uniquely identifies the edges.
    Thus, $|dom(X)|=|dom(Y)| = \Theta(\sqrt{n})$ while $|dom(Z)| = \Theta(n)$.
    Notice that for both relations, $DC(XY,XYZ,1), DC(Z,XYZ,1)$ hold.
    I.e., any pair of variables determine the last variable and there is a variable that determines the whole tuple on its own.
    $C_{X,Y,Z}$ is depicted in Figure \ref{fig:star} for $\frac{n}{7}=16$.
    



    \begin{figure}
        \centering
        \includegraphics[width=0.4\textwidth]{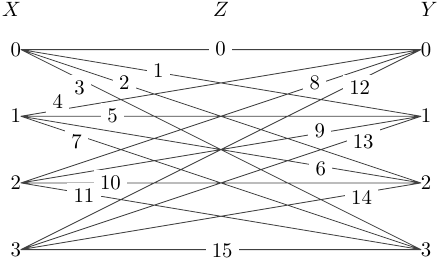}
        \caption{Depiction of $C_{X,Y,Z}$ for $\frac{n}{7}=16$}
        \label{fig:star}
    \end{figure}

    The database instance will be the disjoint union (by renaming constants) of 7 databases.
    4 are of the form $D^R = (R_1^R, R_2^R, R_3^R, R_4^R), R\in \{R_1, R_2, R_3, R_4\}$, and 3 are of the form $D^{XY} = (R_1^{XY}, R_2^{XY}, R_3^{XY},$ $ R_4^{XY}), XY\in \{AU, BV, CW\}$.
    For $D^{R_i}$, We define $R_i^{R_i} = P_{var(R_i)}$ and for $j\neq i$ we set $R_j^{R_i} = C_{var(R_j)\setminus var(R_i), var(R_j)\cap var(R_i)}$.
    Thus, e.g., $R_1^{R_1}=R_{A,W,B}$ and $R_2^{R_1}=C_{C,U,B}$
    
    Furthermore, for $D^{XY}$, we define $R_i^{XY} = C_{var(R_i)\setminus \{X,Y\}, var(R_i)\cap \{XY\}}$.
    Thus, e.g., $R_1^{AU} = C_{B,W,A}$
    Now, by renaming of constants, we can assume the constants of all databases $D^R,D^{XY}$ to be pairwise disjoint.
    In total, $D$ is the disjoint union 
    $$D = \bigcupdot_{R\in \{R_1, R_2, R_3, R_4\}} D^R \cupdot \bigcupdot_{XY \in \{AV, BS, CU\}} D^{XY}.$$

    Clearly, $D\vDash \textbf{PC}$.    
    Now, let $X_1, \dots, X_6$ be an arbitrary variable order for $\hexq$ and the database $D$ as described before.
    Then a VAAT algorithm based on this variable order at least needs to compute the sets $Q^D_{\varhexagon 1} = \Join_i\pi_{X_1}R_i,\dots, Q^D_{\varhexagon 6} = \Join_i\pi_{X_1\cdots X_6}R_i.$
    Importantly, let us consider the set $Q^D_{\varhexagon 4} = \Join_i\pi_{X_1\cdots X_4}R_i$.
    There are two cases:
    
    \textbf{Case 1:} $X_5$ and $X_6$ appear conjointly in a relation.
    Due to the symmetry of the query and the database we can assume,w.l.o.g, $X_5X_6 = UV$ and $$\Join_i\pi_{X_1\cdots X_4}R_i = \Join_i\pi_{ABCW}R_i.$$
    Furthermore, 
    \begin{align*}
        \Join_i\pi_{ABCU}R_i \supseteq & \Join_i\pi_{ABCU}R_i^{R_4} \\
        =& R_1^{R_4} \Join \pi_{BC}R_2^{R_4} \Join \pi_{AC}R_3^{R_4} \Join \pi_{U} R_4^{R_4}\\
        =& C_{A,B,W} \Join \pi_{BC}C_{C,B,U} \Join \pi_{AC} C_{A,C,V} \Join \pi_{W} P_{U,V,W}.
    \end{align*}
    Notice that 
    \begin{align*}
        C_{A,B,W}  = & \{(i_{A},j_{B},(i\sqrt{\frac{n}{7}}+j)_{W}) \mid i,j= 0,\dots, \sqrt{\frac{n}{7}}-1\}\\
        \pi_{BC}C_{C,B,U} =& \{(i_{C},j_{B}) \mid i,j= 0,\dots, \sqrt{\frac{n}{7}}-1\}\\
        \pi_{AC} C_{A,C,V} =& \{(i_{A},j_{C}) \mid i,j= 0,\dots, \sqrt{\frac{n}{7}}-1\} \\
        \pi_{W} P_{U,V,W} = & \{(i\sqrt{\frac{n}{7}}+j)_{W}) \mid i,j= 0,\dots, \sqrt{\frac{n}{7}}-1\}.
    \end{align*}
    Thus,
    $$\Join_i\pi_{ABCU}R_i \supseteq \{i_{A}, j_{B}, k_{C}, (i\sqrt{n}+j)_{U})\mid i,j,k=0,\dots,\sqrt{\frac{n}{7}}-1\}$$
    and 
    $$\max_j|Q^D_{\varhexagon j}| \geq |\Join_i\pi_{ABCU}R_i| = \Omega(n^{1.5}).$$
    
    \textbf{Case 2:} $X_5$ and $X_6$ do not appear conjointly in a relation.
    Due to the symmetry of the query and the database we can assume, w.l.o.g, $X_5X_6 = AU$ and $$\Join_i\pi_{X_1\cdots X_4}R_i = \Join_i\pi_{BCUV}R_i.$$
    Furthermore, 
    \begin{align*}
    \Join_i\pi_{BCUS}R_i \supseteq &\pi_{WB}R_1^{AU} \Join \pi_{BC}R_2^{AU} \Join \pi_{CV}R_3^{AU} \Join \pi_{VW} R_4^{AU}\\
    = & \pi_{WB}C_{W,B,A} \Join \pi_{BC}C_{B,C,U} \Join \pi_{CV}C_{C,V,A} \Join \pi_{VW} C_{V,W,U}.
    \end{align*}
    Notice that 
    \begin{align*}
        \pi_{WB}C_{W,B,A}  = & \{(i_{W},j_{B}) \mid i,j= 0,\dots, \sqrt{\frac{n}{7}}-1\}\\
        \pi_{BC}C_{B,C,U} =& \{(i_{B},j_{C}) \mid i,j= 0,\dots, \sqrt{\frac{n}{7}}-1\}\\
        \pi_{CV}C_{C,V,A} =& \{(i_{C},j_{V}) \mid i,j= 0,\dots, \sqrt{\frac{n}{7}}-1\} \\
        \pi_{VW} C_{V,W,U}= & \{(i_{V},j_{W}) \mid i,j= 0,\dots, \sqrt{\frac{n}{7}}-1\}.
    \end{align*}
    Thus,
    $$\Join_i\pi_{ABCU}R_i \supseteq \{(i_{B}, j_{C}, k_{U}, l_{S})\mid i,j,k,l=0,\dots,\sqrt{\frac{n}{7}}-1\}$$
    and 
    $$\max_j|Q^D_{\varhexagon j}| \geq |\Join_i\pi_{BCUS}R_i| = \Omega(n^{2}).$$
    This completes the proof.
\end{proof}

\section{Further Details for Section \ref{sec:general}}
\label{app:general}
\thmalgoAprox*
\begin{proof}
    On the one hand, we need to verify that Algorithm \ref{alg:approxdecomp} can be implemented to run in linear time, and, on the other hand, that this process indeed leads to the stated approximation guarantee.

    We start with discussing the runtime.
    To that end, we have to ensure that the while-loop only requires linear time and thus, we must poll the minimum in constant time.
    This is achievable by storing pairs $\textbf{X}\in \mathcal{X},\textbf{x}\in \pi_{\textbf{X}}R$ in a queue with the value determining the position in the queue being $|\pi_\textbf{Y}\sigma_{\textbf{X} = \textbf{x}}R|\neq 0$.
    Ties are broken arbitrarily.
    Using bucket sort, creating the priority queue only requires linear time as $|\pi_\textbf{Y}\sigma_{\textbf{X} = \textbf{x}}R| \leq |R|$.
    We assume sufficient pointers are saved while creating the queue, in particular between neighbors in the queue.
    Thus, note that decreasing the key value of an element by 1 is possible in constant time.
    Then, given a pair $\textbf{X},\textbf{x}$, computing $\sigma_{\textbf{X} = \textbf{x}}R$ is possible in time $O(|\sigma_{\textbf{X} = \textbf{x}}R|)$.
    Furthermore, we claim that in time $O(|\sigma_{\textbf{X} = \textbf{x}}R|)$ we can remove $\sigma_{\textbf{X} = \textbf{x}}R$ from $R$ and maintain the queue.
    Maintaining the queue means 
    \begin{itemize}
        \item removing pairs $\textbf{X}',\textbf{x}'$ from the queue where $\textbf{x}'$ is no longer in $\pi_{\textbf{X}'}(R\setminus \sigma_{\textbf{X} = \textbf{x}}R)$, and
        \item updating the value of pairs $\textbf{X}',\textbf{x}'$ where $|\pi_\textbf{Y}\sigma_{\textbf{X}' = \textbf{x}'}R|\neq |\pi_\textbf{Y}\sigma_{\textbf{X}' = \textbf{x}'}(R\setminus \sigma_{\textbf{X} = \textbf{x}}R)|$.
    \end{itemize}  
            
    To accomplish this, 
    for each $\textbf{X}'\in \mathcal{X}-\textbf{X}$ and $\textbf{x}' \in \pi_{\textbf{X}'}\sigma_{\textbf{X} = \textbf{x}}R$ we decrease the value of $\textbf{X}', \textbf{x}'$ by $|\sigma_{\textbf{X}' = \textbf{x}'}\sigma_{\textbf{X} = \textbf{x}}R|$ in the queue (in total, this requires at most $|\mathcal{X}||\sigma_{\textbf{X}=\textbf{x}}R|$ decreases of a key value by 1).
    Notice that, $\textbf{x}'$ no longer being in $\pi_{\textbf{X}'}(R\setminus \sigma_{\textbf{X} = \textbf{x}}R)$ happens exactly when $|\pi_\textbf{Y}\sigma_{\textbf{X}' = \textbf{x}'}R|\neq |\pi_\textbf{Y}\sigma_{\textbf{X}' = \textbf{x}'}(R\setminus \sigma_{\textbf{X} = \textbf{x}}R)| = 0$.
    Hence, if the value of a tuple drops to 0, the corresponding tuple can be removed from the queue. 
    Hence, as we are considering data complexity, Algorithm \ref{alg:approxdecomp} requires linear time in total.
    
    For the correctness of the approximation assume towards a contradiction that this is not the case.
    To that end, let $\bigcupdot_{\textbf{X} \in \mathcal{X}}R^\textbf{X} = R$ be the partitioning created by Algorithm \ref{alg:approxdecomp} and $\textbf{X}'\in \mathcal{X},\textbf{x}'\in \pi_{\textbf{X}'}R$ be the first pair $(\textbf{X}',\textbf{x}')$ such that $|\pi_\textbf{Y}\sigma_{\textbf{X}'=\textbf{x}'}R^{\textbf{X}'}|> |\mathcal{X}|d$.
    Note that the tuples $|\pi_\textbf{Y}\sigma_{\textbf{X}'=\textbf{x}'}R^{\textbf{X}'}|$ all have to be added to $R^{\textbf{X}'}$ at the same time.
    Let us consider the state of $R$ in the algorithm right before they are removed from $R$, which we will denote by $R_\mathcal{A}$.
    Hence, $\textbf{X}',\textbf{x}'=\argmin_{\textbf{X}\in \mathcal{X},\textbf{x}\in \pi_\textbf{X}R_\mathcal{A}}|\pi_\textbf{Y}\sigma_{\textbf{X} = \textbf{x}}R_\mathcal{A}|$ and, consequently, $|\pi_\textbf{Y}\sigma_{\textbf{X} = \textbf{x}}R_\mathcal{A}| > |\mathcal{X}|d$ for all $\textbf{X}\in \mathcal{X}, \textbf{x}\in \pi_{\textbf{X}}R_\mathcal{A}$.
    However, since $R_\mathcal{A} \subseteq R$, there exists an optimal partition $\bigcupdot_{\textbf{X}\in \mathcal{X}}R_\mathcal{A}^\textbf{X} = R_\mathcal{A}$ such that $|\pi_\textbf{Y}\sigma_{\textbf{X} = \textbf{x}}R_\mathcal{A}^\textbf{X}| \leq d$ holds for all $\textbf{X}\in \mathcal{X}, \textbf{x}\in \pi_\textbf{X}R_\mathcal{A}$.
    Therefore, on the one hand, $|\mathcal{X}|\cdot|\pi_\textbf{Y}R_\mathcal{A}| = \sum_{\textbf{X}\in \mathcal{X}} \sum_{\textbf{x} \in \pi_\textbf{X}R_\mathcal{A}}|\pi_\textbf{Y}\sigma_{\textbf{X} = \textbf{x}}R_\mathcal{A}| > \sum_{\textbf{X}\in \mathcal{X}} \sum_{\textbf{x} \in \pi_\textbf{X}R_\mathcal{A}}|\mathcal{X}|d$ and, on the other hand, $|\pi_\textbf{Y}R_\mathcal{A}| \leq  \sum_{\textbf{X}\in \mathcal{X}}\sum_{\textbf{x}\in \pi_\textbf{X}R_\mathcal{A}}|\pi_\textbf{Y}\sigma_{\textbf{X} = \textbf{x}}R_\mathcal{A}^\textbf{X}| \leq \sum_{\textbf{X}\in \mathcal{X}}\sum_{\textbf{x}\in \pi_\textbf{X}R_\mathcal{A}}d$.
    Both cannot be true at the same time.
\end{proof}

\thmalgoExact*
\begin{proof} 
We start by arguing the correctness of the algorithm.
To that end, Algorithm \ref{alg:exactdecomp} maintains that $(R^\textbf{X})_{\textbf{X}\in \mathcal{X}}$ is an optimal decomposition of $R_{\text{dec}} :=\bigcupdot_{\textbf{X}\in \mathcal{X}}R^\textbf{X}$ with $\GDC_{R_{\text{dec}}}(\mathcal{X},\textbf{Y})=d = \max_{\textbf{X}\in \mathcal{X}}DC_{R^\textbf{X}}(\textbf{X},\textbf{Y})$ as a loop invariant (for Lines 5-14).
Additionally, $R_{\text{dec}}\cupdot R_{\text{cur}} = R_{\text{ori}}$, where $R_{\text{cur}}$ and $ R_{\text{ori}}$ respectively are the current and original value of $R$.

Of course, the loop invariant holds before the first loop iteration.
Then, the algorithm iterative checks the existence of an augmenting path to allocate the next $\textbf{y}_0\in \pi_\textbf{Y}R$ and, if it exists, computes a shortest one $(\textbf{y}_1,\dots,\textbf{y}_m,\textbf{X}_1,\dots,\textbf{X}_{m+1})$.
In that case, it cascadingly reallocates $\textbf{y}_1,\dots,\textbf{y}_m$ to $\textbf{X}_2,\dots,\textbf{X}_{m+1}$ and newly allocates $\textbf{y}_0$ to $\textbf{X}_1$.
As alluded to before, this can only result in an increase of the DC for $R^{\textbf{X}_{m+1}}$.
However, Property \ref{prop:dcleq} of augmenting paths ensures that this does not increase the overall maximum $\max_{\textbf{X}\in \mathcal{X}}DC_{R^\textbf{X}}(\textbf{X},\textbf{Y})$.

Thus, it only remains to argue that there always exists an augmenting path if the \GDC of $R_{\text{dec}} \cup \sigma_{\textbf{Y} = \textbf{y}_0}R$ is the same as the \GDC of $R_{\text{dec}}$.
Hence, if no augmenting path can be found, it is justified to increase $d$ by 1 and it does not matter where $\sigma_{\textbf{Y} = \textbf{y}_0}R$ is allocated to.

To that end, let $(R_{\text{opt}}^\textbf{X})_{\textbf{X}\in \mathcal{X}}$ be an optimal partitioning of $R_{\text{dec}} \cup \sigma_{\textbf{Y} = \textbf{y}_0}R$.
Moreover, let $\max_\textbf{X} DC_{R_\text{opt}^\textbf{X}}(\textbf{X},\textbf{Y}) = \max_\textbf{X} DC_{R^\textbf{X}}(\textbf{X},\textbf{Y})$.
Intuitively, the optimal partitioning $(R_{\text{opt}}^\textbf{X})_{\textbf{X}\in \mathcal{X}}$ tells us where the tuples belong and leads us to an augmenting path.
We start by setting $\textbf{X}_1$ such that $\textbf{y}_0\in \pi_\textbf{Y}R_\text{opt}^{\textbf{X}_1}$.
Then, either $(\textbf{X}_1)$ is an augmenting path (w.r.t.\ $(R^\textbf{X})_{\textbf{X}\in \mathcal{X}}$) or
Property \ref{prop:dcleq} is not satisfied.
In the latter case, there must be at least as many elements in 
$|\pi_{\textbf{Y}}\sigma_{\textbf{X}_1 = \textbf{y}_0}R^{\textbf{X}_1}| $ as in $|\pi_{\textbf{Y}}\sigma_{\textbf{X}_1 = \textbf{y}_0}R_\text{opt}^{\textbf{X}_1}|$.
However, $\pi_{\textbf{Y}}R_\text{opt}^{\textbf{X}_1}$ contains $\textbf{y}_0$ and thus there is a ``misplaced'' element  $\textbf{y}_1\in \pi_{\textbf{Y}}\sigma_{\textbf{X}_1 = \textbf{y}_0}R^{\textbf{X}_1} \setminus \pi_{\textbf{Y}}\sigma_{\textbf{X}_1 = \textbf{y}_0}R_\text{opt}^{\textbf{X}_1}$.
We can then set $\textbf{X}_2$ such that $\textbf{y}_1\in \pi_\textbf{Y}R_\text{opt}^{\textbf{X}_2}$.

Due to a similar argumentation as before, either $(\textbf{y}_1,\textbf{X}_1, \textbf{X}_2)$ is an augmenting path or there must exist a misplaced $\textbf{y}_2\in \pi_{\textbf{Y}}\sigma_{\textbf{X}_2 = \textbf{y}_1}R^{\textbf{X}_2} \setminus \pi_{\textbf{Y}}\sigma_{\textbf{X}_2 = \textbf{y}_1}R_\text{opt}^{\textbf{X}_2}$ and we can set $\textbf{X}_3$ such that $\textbf{y}_2\in \pi_\textbf{Y}R_\text{opt}^{\textbf{X}_3}$.

We proceed inductively: Let $(\textbf{y}_1,\dots, \textbf{y}_i,\textbf{X}_1, \dots, \textbf{X}_{i+1})$ be such that 
\begin{enumerate}
    \item For all $j\in \{0,\dots,i\}$ we have $\textbf{y}_j\in \pi_{\textbf{Y}}R^{\textbf{X}_{j+1}}_{\text{opt}}$.
    \item For all $j\in \{1,\dots,i\}$ we have $\textbf{y}_j\in \pi_{\textbf{Y}}\sigma_{\textbf{X}_j=\textbf{y}_{j-1}}R^{\textbf{X}_j}\setminus( \pi_{\textbf{Y}}\sigma_{\textbf{X}_j=\textbf{y}_{j-1}}R^{\textbf{X}_j}_{\text{opt}} \cup \{\textbf{y}_1,\dots,\textbf{y}_{j-1}\})$.
\end{enumerate}
Then, there are three possibilities: Either $(\textbf{y}_1,\dots, \textbf{y}_i,\textbf{X}_1, \dots, \textbf{X}_{i+1})$ is an augmenting path.
In that case we are done.
Or there is a further misplaced element $\textbf{y}_{j+1}\in \pi_{\textbf{Y}}\sigma_{\textbf{X}_{j+1}=\textbf{y}_{j}}R^{\textbf{X}_{j+1}}\setminus( \pi_{\textbf{Y}}\sigma_{\textbf{X}_{j+1}=\textbf{y}_{j}}R^{\textbf{X}_{j+1}}_{\text{opt}} \cup \{\textbf{y}_1,\dots,\textbf{y}_{j}\})$.
In that case, simply consider $(\textbf{y}_1,\dots, \textbf{y}_{i+1},\textbf{X}_1, \dots, \textbf{X}_{i+2})$ with $\textbf{y}_{j+1}\in \pi_{\textbf{Y}}R^{\textbf{X}_{j+2}}_{\text{opt}}$.
Note that this case cannot lead to an infinite induction as the $\textbf{y}_j$ are all different.
In the last case, $\pi_{\textbf{Y}}\sigma_{\textbf{X}_{j+1}=\textbf{y}_{j}}R^{\textbf{X}_{j+1}}\setminus( \pi_{\textbf{Y}}\sigma_{\textbf{X}_{j+1}=\textbf{y}_{j}}R^{\textbf{X}_{j+1}}_{\text{opt}} \cup \{\textbf{y}_1,\dots,\textbf{y}_{j}\})$ is empty.
However, this is not possible as $(\textbf{y}_1,\dots, \textbf{y}_i,\textbf{X}_1, \dots, \textbf{X}_{i+1})$ would then be an augmenting path.
This is since reallocating along $(\textbf{y}_1,\dots, \textbf{y}_i,\textbf{X}_1, \dots, \textbf{X}_{i+1})$ can only increase the DC of $R^{\textbf{X}_{i+1}}$ due to the size of $\pi_\textbf{Y}\sigma_{\textbf{X}_{j+1}=\textbf{y}_{j}}R^{\textbf{X}_{i+1}}$ and, furthermore, there being no more misplaced element implies that after reallocation, $\pi_\textbf{Y}\sigma_{\textbf{X}_{j+1}=\textbf{y}_{j}}R^{\textbf{X}_{i+1}}$ is a subset of $\pi_{\textbf{Y}}\sigma_{\textbf{X}_{j+1}=\textbf{y}_{j}}R^{\textbf{X}_{j+1}}_{\text{opt}}$.


With regard to the time complexity, the main part that merits discussion is the computation of augmenting paths.
For this, however, simply keep pointers from every $\mathbf{y}\in \pi_\textbf{Y}R$ to $\sigma_{\textbf{X} = \textbf{y}}R^\textbf{X}$ for all $\textbf{X}\in \mathcal{X}$.
Then, when searching for an augmenting path for $\mathbf{y}_0\in \pi_\textbf{Y}R$, we simply have to do a breadth-first search, i.e., start with $\mathbf{y}_0$, then move on to all $\bigcup_{\textbf{X}\in \mathcal{X}}\sigma_{\textbf{X} = \textbf{y}_0}R^\textbf{X}$ and follow their pointers if necessary.
This only requires linear time (per $\mathbf{y}_0\in \pi_\textbf{Y}R$) as we need to consider each pointer at most once and each tuple only has a constant number of pointers.
Furthermore, the creation of the pointers can be done once in a preprocessing step.
\end{proof}